\documentclass[letterpaper, 11pt]{article}

\usepackage{mathtools}
\usepackage{enumerate}
\usepackage{amsfonts}
\usepackage{style}
\usepackage{graphicx} 
\usepackage{dsfont}
\usepackage{tipa}
\usepackage{xcolor}
\usepackage{tikz}
\usetikzlibrary{arrows.meta,positioning,calc,decorations.pathreplacing,angles,quotes}
\definecolor{revpurple}{RGB}{128,0,128}
\newcommand{\rev}[1]{\textcolor{black}{#1}}

\newcommand{\FigEquallySpacedPoints}{%
\begin{tikzpicture}[thick, line cap=round, line join=round, scale=0.85]

\def\Rfig{3.5}

\coordinate (Xmin) at ({-\Rfig - 0.8}, 0);
\coordinate (Xmax) at ({\Rfig + 1.0}, 0);
\draw [->, black, thin] (Xmin) -- (Xmax) node [right, black] {$\textsf{Re}$};

\coordinate (Center) at (0,0);

\draw [black, thin] (Center) circle[radius=\Rfig cm];

\filldraw[black] (Center) circle (0.5pt);
\node[below, font=\scriptsize] at (0, -0.15) {$\delta$};

\def\Jhalf{43}
\draw [red, very thick] (Center) ++({-\Jhalf}:\Rfig cm)
  arc[radius=\Rfig cm, start angle={-\Jhalf}, end angle=\Jhalf];
\node[red, font=\scriptsize] at ({0}:{\Rfig + 0.2}) {$J$};

\def\ABhalf{40}
\filldraw[black] (\ABhalf:\Rfig cm) circle (0.8pt);
\node[above, font=\scriptsize, xshift=2pt] at (\ABhalf:\Rfig cm) {$\alpha$};
\filldraw[black] ({-\ABhalf}:\Rfig cm) circle (0.8pt);
\node[below right, font=\scriptsize, xshift=-2pt] at ({-\ABhalf}:\Rfig cm) {$\beta$};

\filldraw[black] ({\Rfig},0) circle (0.8pt);
\node[below, font=\scriptsize, yshift=-2pt] at ({\Rfig}, 0) {$1$};

\filldraw[violet] (0:\Rfig cm) circle (1.8pt);
\node[blue!80!black, below right, font=\scriptsize\bfseries, yshift=-1pt, xshift=1pt] at (0:\Rfig cm) {$\eta_0$};

\filldraw[violet] (36:\Rfig cm) circle (1.8pt);
\node[blue!80!black, right, font=\scriptsize] at ($(36:\Rfig cm) + (0.08, 0)$) {$\eta_1$};
\draw[blue!80!black, thin] ($(36:\Rfig cm) + (126:0.1cm)$) -- ($(36:\Rfig cm) + (-54:0.1cm)$);

\filldraw[violet] (72:\Rfig cm) circle (1.8pt);
\node[blue!80!black, above, font=\scriptsize, yshift=2pt] at (72:\Rfig cm) {$\eta_2$};
\draw[blue!80!black, thin] ($(72:\Rfig cm) + (162:0.1cm)$) -- ($(72:\Rfig cm) + (-18:0.1cm)$);

\filldraw[violet] (108:\Rfig cm) circle (1.8pt);
\node[blue!80!black, above left, font=\scriptsize] at (108:\Rfig cm) {$\eta_3$};
\draw[blue!80!black, thin] ($(108:\Rfig cm) + (198:0.1cm)$) -- ($(108:\Rfig cm) + (18:0.1cm)$);

\filldraw[violet] (144:\Rfig cm) circle (1.8pt);
\node[blue!80!black, left, font=\scriptsize] at (144:\Rfig cm) {$\eta_4$};
\draw[blue!80!black, thin] ($(144:\Rfig cm) + (234:0.1cm)$) -- ($(144:\Rfig cm) + (54:0.1cm)$);

\filldraw[violet] (-36:\Rfig cm) circle (1.8pt);
\node[blue!80!black, right, font=\scriptsize] at ($(-36:\Rfig cm) + (0.08, 0)$) {$\eta_{-1}$};
\draw[blue!80!black, thin] ($(-36:\Rfig cm) + (-126:0.1cm)$) -- ($(-36:\Rfig cm) + (54:0.1cm)$);

\filldraw[violet] (-72:\Rfig cm) circle (1.8pt);
\node[blue!80!black, below, font=\scriptsize, yshift=-2pt] at (-72:\Rfig cm) {$\eta_{-2}$};
\draw[blue!80!black, thin] ($(-72:\Rfig cm) + (-162:0.1cm)$) -- ($(-72:\Rfig cm) + (18:0.1cm)$);

\filldraw[violet] (-108:\Rfig cm) circle (1.8pt);
\node[blue!80!black, below left, font=\scriptsize] at (-108:\Rfig cm) {$\eta_{-3}$};
\draw[blue!80!black, thin] ($(-108:\Rfig cm) + (-198:0.1cm)$) -- ($(-108:\Rfig cm) + (-18:0.1cm)$);

\filldraw[violet] (-144:\Rfig cm) circle (1.8pt);
\node[blue!80!black, left, font=\scriptsize] at (-144:\Rfig cm) {$\eta_{-4}$};
\draw[blue!80!black, thin] ($(-144:\Rfig cm) + (-234:0.1cm)$) -- ($(-144:\Rfig cm) + (-54:0.1cm)$);

\filldraw[violet] (180:\Rfig cm) circle (1.8pt);
\node[blue!80!black, left, font=\scriptsize] at ($(180:\Rfig cm) + (-0.08, 0)$) {$\eta_5$};

\end{tikzpicture}%
}

\newcommand{\FigGeneralizedLemma}{%
\begin{tikzpicture}[thick, line cap=round, line join=round, scale=0.85]

\def\Rfig{3.5}
\def\Rsmall{2.8}
\def\ABhalf{40}

\coordinate (Xmin) at ({-\Rfig - 0.8}, 0);
\coordinate (Xmax) at ({\Rfig + 1.0}, 0);
\draw [->, black, thin] (Xmin) -- (Xmax) node [right, black] {$\textsf{Re}$};

\coordinate (Center) at (0,0);

\draw [black, thin] (Center) circle[radius=\Rfig cm];
\node[black, font=\scriptsize] at (120:{\Rfig + 0.25}) {$\Gamma$};

\draw [black, thick, densely dashed]
  (Center) ++({-\ABhalf}:\Rfig cm)
  arc[radius=\Rfig cm, start angle={-\ABhalf}, end angle=\ABhalf];

\draw [blue!60!black, thin] (Center) circle[radius=\Rsmall cm];
\node[blue!60!black, font=\scriptsize] at ({200}:{\Rsmall + 0.3}) {$\Gamma_\mu$};

\filldraw[black] (Center) circle (0.5pt);
\node[below, font=\scriptsize, xshift=-2pt, yshift=-2pt] at (Center) {$\delta$};

\draw [red, very thick] (Center) ++({-\ABhalf}:\Rsmall cm)
  arc[radius=\Rsmall cm, start angle={-\ABhalf}, end angle=\ABhalf];
\node[red, font=\scriptsize, left, xshift=-2pt, yshift=3pt] at ({0}:{\Rsmall cm}) {$P$};

\filldraw[black] ({\Rsmall}, 0) circle (0.8pt);
\node[below, font=\scriptsize] at ({\Rsmall}, -0.15) {$1{-}\delta{-}\mu$};

\filldraw[black] (\ABhalf:\Rsmall cm) circle (0.8pt);
\node[above right, font=\scriptsize, xshift=-3pt] at (\ABhalf:\Rsmall cm) {$\alpha$};
\filldraw[black] ({-\ABhalf}:\Rsmall cm) circle (0.8pt);
\node[below right, font=\scriptsize, xshift=-2pt, yshift=-2pt] at ({-\ABhalf}:\Rsmall cm) {$\beta$};

\filldraw[black] ({\Rfig}, 0) circle (0.8pt);
\node[above right, font=\scriptsize] at ({\Rfig + 0.05}, 0.05) {$1$};

\end{tikzpicture}%
}

\newcommand{\rprep}{r_{\textnormal{prep}}}
\newcommand{\rmeas}{r_{\textnormal{meas}}}
\renewcommand{\nu}{\mathfrak{r}}

\begin{document}

\title{SPAM Tolerance for Pauli Error Estimation}
\author{Ryan O'Donnell \thanks{Carnegie Mellon University, \texttt{odonnell@cs.cmu.edu}} \and Samvitti Sharma \thanks{Carnegie Mellon University, \texttt{samvitts@cmu.edu}}}

\maketitle

\begin{abstract}
The Pauli channel is a fundamental model of noise in quantum systems, motivating the task of Pauli error estimation. We present an algorithm that builds on the reduction to Population Recovery introduced in \cite{FO21paulierror}.  Addressing an open question from that work, our algorithm has the key advantage of robustness against even severe state preparation and measurement (SPAM) errors.  To tolerate SPAM, we must analyze Population Recovery on a combined \rev{$Z$-channel}/bit-flip channel, which necessitates extending the complex analysis techniques from \cite{polyanskiy2017sample,DOS17populationrecovery}.  For $n$-qubit channels, our Pauli error estimation algorithm requires only $\exp(n^{1/3})$ unentangled state preparations and measurements, improving on previous SPAM-tolerant algorithms that had $2^n$-dependence even for restricted families of Pauli channels.  We also give evidence that no SPAM-tolerant method can make asymptotically fewer than $\exp(n^{1/3})$ uses of the channel.
\end{abstract}

\section{Introduction}
One of the most widely used theoretical models for noise in experimental quantum systems is the $n$-qubit \emph{Pauli channel}, a mixed-unitary channel that applies $n$-qubit Pauli operators to an input $n$-qubit quantum state according to a fixed probability distribution, referred to as the \emph{Pauli error rates}. \emph{Pauli error estimation} is the task of learning the Pauli error rates of an unknown channel to precision $\epsilon$ in some metric. Since there are $4^n$ Pauli error rates, yet we seek subexponential complexity, a natural choice for the precision metric is~$\ell_\infty$, as this means we need only find and estimate the largest $1/\eps$ Pauli error rates. 

Pauli error estimation plays a central role in diagnosing and modeling errors in present-day experimental quantum setups. As such, we are practically motivated to design simple algorithms that do not involve entangling operations; i.e., algorithms that prepare unentangled states, pass them through the channel, and perform unentangled measurements (followed by classical postprocessing).  
Moreover, for relevance to real-world quantum implementations, it has long been considered  desirable~\cite{knill2008randomized} to develop algorithms that are robust against \emph{state preparation and measurement} (SPAM) errors.  See, e.g., \cite{chen2023learnability} for more on practical aspects of Pauli channel estimation in the presence of SPAM.

\subsection{Our result}
In this work, we give the first entanglement-free Pauli error estimation algorithm that is robust to \rev{SPAM errors modeled as depolarizing noise (see \Cref{sec:spam} for a precise definition). When SPAM is severe, the sample complexity is subexponential in~$n$, improving on previous SPAM-tolerant algorithms that required exponentially many samples; when SPAM is mild, the complexity reduces to $\poly(n/\eps)$, essentially matching the SPAM-free setting.}  Here we briefly describe the setting (see \Cref{sec:prelims,sec:spam} for more complete definitions), and then state our main theorem.

\paragraph{Setup: the channel.}  We assume access to an $n$-qubit Pauli channel $\mathcal{E}$ with unknown error rates~$\pi$.  That is, $\pi$ is a probability distribution on $\{I, X, Y, Z\}^n$, and $\mathcal{E}$ may be modeled as drawing $\bP \sim \pi$ and then applying the $n$-qubit Pauli unitary~$\bP$.\footnote{We use \textbf{boldface} to denote random variables.}  In fact, for a completely general $n$-qubit channel, there is a natural notion of its ``Pauli error rates'', arising from the technique of randomized compiling~\cite{wallman2016noise}.  Our algorithm can easily be extended to learn Pauli error rates in this more general setting too; see \Cref{rem:rc}.

\paragraph{Setup: SPAM.}  We only consider algorithms that employ single-qubit state preparation and single-qubit measurements.  Moreover, we assume these are subject to SPAM errors, governed by  ``retention parameters'' $\rprep,\; \rmeas \in [0,1]$, known to the algorithm (via prior calibration). As discussed in \Cref{sec:spam}, several different natural operational interpretations of SPAM are all mathematically equivalent to the following model: after a state is prepared, it is passed through the depolarizing channel \rev{$\rho \mapsto \rprep \cdot \rho + (1 - \rprep)\cdot\frac{1}{2}I$}; and, before it is measured, it is passed through the channel \rev{$\rho \mapsto \rmeas \cdot \rho + (1 - \rmeas)\cdot\frac{1}{2}I$}.
In particular, $\rprep = \rmeas = 1$ corresponds to no SPAM.

\bigskip

We may now state our main result:

\begin{theorem}\label{thm:mainresult}
    In the above setup, let $r = \rprep \cdot \rmeas$ be the overall SPAM retention parameter, and write $r = 1 - \delta$, which may be thought of as the overall rate of SPAM.
    Let a precision parameter $0 < \eps \leq 1/2$ be given, and assume $\delta \leq 0.99$.
    Then there is an algorithm with the following properties:
    \begin{itemize}
        \item For $m = \begin{cases}
            \rev{\poly(n/\eps) \cdot}\exp\left(O\left((\delta n)^{1/3}\ln^{2/3}(1/\epsilon)\right)\right) & \text{if } \delta \geq 
            \rev{\delta_c } \\
            \poly(n/\eps)
            &\text{if } \delta < \rev{\delta_c }
        \end{cases}$, \rev{where $\delta_c = \frac{6}{\pi^3} \cdot \frac{\ln(1/\epsilon)}{n}$}, it prepares $m$ unentangled $n$-qubit states, where each qubit is chosen uniformly at random from $\{\ket{0}, \ket{i}, \ket{+}\}$.
        
        \item After passing these through the channel, it performs unentangled measurements on the resulting qubits, in only the $\{\ket{0}, \ket{1}\}$,  $\{\ket{i}, \ket{-i}\}$, and  $\{\ket{+}, \ket{-}\}$ bases.
        \item After efficient classical postprocessing, it returns an estimate of the Pauli error rates $\Tilde{\pi}$ such that $\Vert \Tilde{\pi} - \pi \Vert_\infty \leq \epsilon$ with high probability.
    \end{itemize} 
\end{theorem}

\paragraph{Algorithm overview.} At a high level, the algorithm proceeds in three stages.

\emph{Stage 1: Data collection.} Alice repeatedly prepares random unentangled $n$-qubit states (with each qubit chosen uniformly from $\{\ket{0}, \ket{i}, \ket{+}\}$), passes them through the channel, and measures each qubit in the corresponding basis (\Cref{def:probe}).

\emph{Stage 2: Reduction to Population Recovery.} Using the classical reduction of~\cite{FO21paulierror} (\Cref{sec:recap}), each such measurement round produces a noisy binary string that can be viewed as a sample from the Pauli error distribution $\pi$ passed through a binary $Z$-channel $\Xi$ (\Cref{def:xi}). Under SPAM (\Cref{sec:spam}), this becomes the combination of a binary symmetric channel and a binary $Z$-channel, denoted as $\textrm{BSC}_{\delta/2} \circ \textrm{Z}_{1/3}$. A branch-and-prune reduction (\Cref{prop:individ}) further reduces the full Population Recovery problem to Individual Recovery: estimating $\pi(B)$ for a given string $B$.

\emph{Stage 3: Individual Recovery.} The Individual Recovery algorithm uses the methodology of \cite{DOS17populationrecovery}, whose sample complexity is controlled by the quantity $\eta(\eps)$ that we lower-bound in \Cref{sec:lowerbounding} via new complex-analytic techniques (\Cref{thm:etafinalbound}).

\begin{figure}[ht]
\centering
\begin{tikzpicture}[
    node distance=0.5cm,
    box/.style={draw, rounded corners, minimum height=0.9cm, minimum width=2cm, align=center, font=\small},
    arr/.style={-{Stealth[length=5pt]}, thick}
]
\node[box] (prep) {Prepare random\\product states};
\node[box, right=of prep] (spam1) {SPAM\\(state prep)};
\node[box, right=of spam1] (channel) {Pauli\\channel};
\node[box, right=of channel] (spam2) {SPAM\\(meas.)};
\node[box, right=of spam2] (meas) {Measure in\\eigenbasis};

\draw[arr] (prep) -- (spam1);
\draw[arr] (spam1) -- (channel);
\draw[arr] (channel) -- (spam2);
\draw[arr] (spam2) -- (meas);

\node[below=0.7cm of channel, font=\small\itshape, align=center] (equiv) {Equivalent to noisy samples through:};

\node[box, below=0.3cm of equiv] (zchan) {$\mathrm{Z}_{1/3}$};
\node[box, right=of zchan] (bsc) {$\mathrm{BSC}_{\delta/2}$};
\node[left=0.3cm of zchan, font=\small] {Channel:};
\draw[arr] (zchan) -- (bsc);

\node[box, below=0.7cm of zchan, minimum width=4.5cm] (indiv) {Individual Recovery};
\node[box, below=0.5cm of indiv, minimum width=4.5cm] (output) {Branch-and-prune};

\draw[arr] (bsc.south) -- ++(0,-0.25) -| (indiv.north);
\draw[arr] (indiv) -- (output);
\end{tikzpicture}
\caption{Schematic of the algorithm pipeline.}
\label{fig:pipeline}
\end{figure}

Notably, even in the presence of \textit{low} SPAM error, which corresponds to the case where \rev{$\delta < \delta_c$}, our algorithm achieves sample complexity and runtime $\poly(n/\eps)$. This matches, up to polynomial factors, the bounds previously established for the SPAM-free setting. The algorithm presented in \cite{FO21paulierror}, which assumes no SPAM error, has overall runtime $O(n \log(n/\epsilon) \cdot \epsilon^{-3})$ and sample complexity $m = O(\epsilon^{-2})\log(n/\epsilon)$.

We remark that although $m$ is subexponential in $n$, a dependence of $\exp(n^{1/3})$ seems reasonably plausible in real-world settings with $50 \leq n \leq 1000$, whereas prior work's dependence of~$2^{\Omega(n)}$ is infeasible in such settings.  Moreover, as we discuss below, there is strong evidence that \emph{any} Pauli error estimation algorithm that tolerates even a tiny amount of SPAM must use the channel on the order of $\exp(n^{1/3})$ times.

\subsection{The connection with Population Recovery, and the difficulty of SPAM}

Let us first recall how the Pauli error estimation problem would be solved if: (i)~entangling operations were allowed; (ii)~perfect state preparation and measurement is assumed.  In this case, there is a simple algorithm based on superdense coding~\cite{bennett1992communication}:  Prepare $n$ Bell pairs in registers $A_1B_1, \dots, A_n B_n$, pass $A_1 \cdots A_n$ through the channel, and then measure each $A_iB_i$ pair in the Bell basis. It is easy to show that if the channel applies $n$-bit Pauli operator $\bP$ (which it does with probability~$\pi(\bP)$), then $\bP \in \{I, X, Y, Z\}^n$ can be perfectly ``read off'' from the measurement outcomes.
Thus with entangling operations allowed and no SPAM, the problem reduces to a classical task: learning an unknown probability distribution on~$[4]^n$ to $\ell_\infty$-precision~$\eps$, which is easily done with $O(1/\eps^2)$ samples~\cite{canonne2020short}.


Suppose now that we consider a more realistic setting, with a slight amount of state preparation or measurement noise (or both).  In this setting, \emph{even if we allow entangling operations}, the problem becomes much more difficult. For example, suppose we are even \emph{promised} that the Pauli channel is only supported on $\{I, X\}^n$ Pauli strings.  In this case, the superdense coding strategy is unnecessary: without SPAM, we could simply repeatedly pass $\ket{0}^{\otimes n}$ through the channel, measure in the $\{\ket{0}, \ket{1}\}$ basis, and be reduced to the classical task of learning an unknown probability distribution on $\{\textsf{0},\textsf{1}\}^n$.  But suppose there is some small constant probability $p > 0$ that each measurement outcome is flipped (with a similar observation being possible for state preparation error).  Then we precisely have an instance of the classical \emph{Population Recovery} problem~\cite{dvir2012restriction,wigderson2016population}: learning a probability distribution on $\{0,1\}^n$ when the samples are passed through a binary symmetric channel with bit-flip probability~$p$.  For this task it is known~\cite{polyanskiy2017sample,DOS17populationrecovery} that \emph{any} algorithm that succeeds must use at least $\exp(\Omega(n^{1/3} \ln^{2/3}(1/\eps))$ samples.  Thus even in this essentially classical version of the Pauli error estimation problem, any slight positive rate of SPAM error means that one cannot asymptotically improve on the ``$m$'' in our \Cref{thm:mainresult}.

\subsection{Prior work, and our improvements}

A main prior work on SPAM-tolerant Pauli error estimation is by Flammia and Wallman~\cite{FlammiaWallman}.  Among other results, they gave a SPAM-tolerant Pauli error estimation algorithm that succeeds in the harder task of estimating the Fourier spectrum of the Pauli error rates.  Its total complexity, however is $\Tilde{O}(\frac{1}{\Delta} \cdot 2^n/\eps^2)$, where $\Delta$ is the ``spectral gap''.  In addition to having fully exponential complexity in~$n$, the dependence on~$\Delta$ may be unfavorable, as this quantity may be arbitrarily small, or even $0$ for some simple channels.
Some improvements to this result were made in~\cite{chen2022quantum}, at the cost of adding entanglement with $n$ ancilla qubits into the protocol.

\rev{Another relevant prior work is by Harper, Yu, and Flammia~\cite{HYF21}, who proposed a SPAM-robust Pauli estimation protocol under a specific noise model. Their approach assumes that the support of the Pauli channel is chosen uniformly at random, and measures sample complexity in terms of queries to an ``eigenvalue oracle'' with Gaussian noise, which can only heuristically be approximated by quantum measurements. By contrast, our protocol makes no assumptions on the support of the Pauli channel and uses standard quantum state preparations and measurements. The complexity scaling also differs: their method achieves polynomial complexity under their specific assumptions, whereas our $\exp(n^{1/3})$ scaling applies to the fully general (depolarizing) SPAM setting.}

By contrast, our approach builds on the subsequent paper by Flammia and O'Donnell~\cite{FO21paulierror}, which was the first to make a connection between Pauli error estimation and Population Recovery.  This work gave a  novel method, using \emph{no} entangling operations, that in the SPAM-free setting reduced Pauli error estimation to another form of Population Recovery: learning a probability distribution on $\{\textsf{0},\textsf{1}\}^n$ in which the samples are passed to the learner through a ``Z-channel'' with crossover probability~$\frac13$ (see \Cref{def:Z}). Z-channels had not previously been studied in the context of Population Recovery, but it was observed in \cite{FO21paulierror} that prior Population Recovery algorithms for the binary erasure channel  apply equally well for the Z-channel.  By using such prior algorithms~\cite{dvir2012restriction,moitra2013polynomial}, Flammia--O'Donnell obtained the same result \Cref{thm:mainresult} but with two differences:
\begin{itemize}
    \item they used the channel only $m = O(\log(n/\epsilon)/\epsilon^2)$ times;
    \item however, they assumed no SPAM errors.\footnote{Actually, \cite{FO21paulierror} were able to tolerate a very limited form of ``measurement error''; specifically, one where each measurement is either perfect, or else reports ``error'' with probability at most~$\frac14$.  They showed that such special measurement errors had the effect of increasing the Z-channel crossover probability up to~$\frac12$, the threshold below which prior BEC Population Recovery algorithms worked just as efficiently.} 
\end{itemize}

In this work, we first show that for the very general SPAM model we allow in \Cref{thm:mainresult}, the Pauli error estimation problem can be reduced to yet another variant of Population Recovery: namely, one that \emph{combines} $\frac13$-rate Z-channel noise with $\delta$-rate BSC noise.  This kind of channel is more general than all the ones previously studied for Population Recovery.  Then, the main effort is to give a Population Recovery algorithm for this general channel; for this, we need to extend the approach from~\cite{polyanskiy2017sample,DOS17populationrecovery}, which in turn requires generalizing certain complex analysis theorems from~\cite{BE97}.

\section{Preliminaries} \label{sec:prelims}
\subsection{Quantum preliminaries}

\begin{definition}
The four \textit{1-qubit Pauli operators} are defined as  
\begin{equation}
    \sigma_0 = \begin{pmatrix}
    1 & 0 \\ 
    0 & 1 
\end{pmatrix}, \quad \sigma_1 = \sigma_x= \begin{pmatrix}
    0 & 1 \\ 
    1 & 0 
\end{pmatrix}, \quad \sigma_2 = \sigma_y= \begin{pmatrix}
    0 & -i \\ 
    i & 0 
\end{pmatrix}, \quad \sigma_3 = \sigma_z= \begin{pmatrix}
    1 & 0 \\ 
    0 & -1 
\end{pmatrix},
\end{equation}
with $\sigma_0$ being  the identity operator, and $\sigma_1, \sigma_2$, and $\sigma_3$ corresponding to  rotations on the Bloch sphere by $180^\circ$ degrees about the $x$- $y$-, and $z$-axes, respectively. More generally, an \emph{$n$-qubit Pauli operator} is defined as $\sigma_{C} = \bigotimes_{i=1}^n \sigma_{C_i}$, where $C$ denotes a string in $\{0, 1, 2, 3\}^n$. 
\end{definition}
\begin{definition}
    We define the operation $\oplus$ on $\{0,1,2,3\}$ by $i \oplus j = k$ iff $\sigma_i \sigma_j = \sigma_k$ (up to phase).  Equivalently, $\oplus$ is bitwise xor when $\{0,1,2,3\}$ are regarded as $2$-bit base-$2$ numbers.  We extend $\oplus$ to operate coordinatewise on $\{0,1,2,3\}^n$.
\end{definition}

\begin{definition}
An \textit{n-qubit Pauli channel} is defined as 
\begin{equation} \rho \mapsto \sum_{C \in \{0, 1, 2, 3\}^n} \pi(C) \cdot \sigma_C\rho\sigma_C^{\dagger}, \end{equation}
where $\rho$ denotes the input $n$-qubit quantum state, and $\pi$ is a probability distribution on $\{0, 1, 2, 3\}^n$. We refer to the parameters $\pi(C)$ as the \emph{Pauli error rates}.
\end{definition}

\begin{definition}
The orthonormal eigenbasis for the Pauli operator $\sigma_1 = \sigma_x$ will be denoted by $\ket{\chi_+^1}, \ket{\chi_-^1}$. Note that this basis corresponds to the two unit vectors on the Bloch sphere pointing along the positive/negative $x$-axis, denoted by $\ket{+}, \ket{-}$ respectively.
Similarly, the eigenbasis for $\sigma_2 = \sigma_y$ is $\ket{\chi_+^2}, \ket{\chi_-^2}$, aka $\ket{i}, \ket{-i}$; and, the eigenbasis for $\sigma_3 = \sigma_z$ is $\ket{\chi_+^3}, \ket{\chi_-^3}$, aka $\ket{0}, \ket{1}$ respectively. 
\end{definition}



\subsection{Population Recovery and classical channels}
We first recall communication channels, in the sense of classical information theory:
\begin{definition}
    Let $\Sigma, \Gamma$ be finite alphabets.  A \emph{classical channel} is a stochastic map $\Lambda : \Sigma \to \Gamma$.
    We extend to $\Lambda : \Sigma^n \to \Gamma^n$ by letting $\Lambda$ act independently on each symbol.
\end{definition}

The most typically studied examples are the following, both with $\Sigma = \{\textsf{0},\textsf{1}\}$:
\begin{itemize}
    \item the \emph{Binary Symmetric Channel} with crossover probability~$p$ (abbreviation: BSC${}_p$), where  $\Gamma = \{\textsf{0},\textsf{1}\}$ and $\Lambda$ flips the bit with probability~$p$;
    \item the \emph{Binary Erasure Channel} with erasure probability~$p$ (abbreviation: BEC${}_p$), where $\Gamma = \{\textsf{0},\textsf{1}, \textsf{?}\}$ and $\Lambda$ preserves each bit with probability $1-p$ and replaces it by $\textsf{?}$ with probability~$p$.
\end{itemize}

\cite{FO21paulierror} also showed the relevance of the following channel for Pauli error estimation:
\begin{definition} \label{def:Z}
    The \emph{Z-channel} with crossover probability~$p$ (abbreviation: Z$_{p}$) has $\Sigma = \Gamma = \{\textsf{0}, \textsf{1}\}$; it always maps $\textsf{0}$ to $\textsf{0}$, and it maps $\textsf{1}$ to $\textsf{0}$ with probability~$p$.
\end{definition}

Now we recall the classical  problem of learning a distribution from noisy data, introduced in~\cite{dvir2012restriction,wigderson2016population} under the name ``Population Recovery'': 
\begin{definition}
    Let $\Lambda : \Sigma \to \Gamma$ be a classical channel.  In the associated \textit{Population Recovery} problem, the goal is to learn an unknown probability distribution $\calD$ on $\Sigma^n$ to a given $\ell_\infty$ precision of $\eps > 0$ (with high probability).  
    The model is that the learner can get independent samples from $\Gamma^n$ distributed as $\Lambda(\bx)$, where $\bx \sim \calD$.
\end{definition}
Population Recovery has been well studied for the BSC and BEC channels (see, e.g.,~\cite{dvir2012restriction,
wigderson2016population,batman2013finding,moitra2013polynomial,lovett2015improved,lovett2017noisy,de2016noisy,polyanskiy2017sample,DOS17populationrecovery}), as well as for the binary deletion channel.

A simplified version of Population Recovery is the following:
\begin{definition}
    In the \emph{Individual Recovery} variant of Population Recovery, a particular string $b \in \Sigma^n$ is given and the only task is to estimate $\calD(b)$ to additive precision~$\eps$ (with high probability).
\end{definition}
It is well known~\cite{dvir2012restriction,polyanskiy2017sample} that  Population Recovery efficiently reduces to Individual Recovery via a branch-and-prune strategy; e.g., Section 4.2 in \cite{FO21paulierror} shows the following:
\begin{proposition} \label{prop:individ}
    For a given channel on constant-sized alphabets, there is a reduction from Population Recovery to Individual Recovery that loses only an $O(n \log(n/\eps))$ factor in terms of runtime and number of experiments.
\end{proposition}

\subsection{Complex analysis}\label{subsec:complex_analysis}
We will make use of the following classical result from complex analysis in Section~\ref{sec:lowerbounding}.

\begin{fact}[Maximum Modulus Principle]\label{fact:mmp}
Let $U \subseteq \C$ be a bounded open region and let $g$ be a function analytic on $U$ and continuous on the closure $\overline{U}$. Then $|g|$ attains its maximum on the boundary $\partial U$; that is,
\[ \max_{z \in \overline{U}} |g(z)| = \max_{z \in \partial U} |g(z)|. \]
For a reference, see e.g.~\cite[Ch.~VI]{conway1978}.
\end{fact}

\subsection{Recapping \cite{FO21paulierror}} \label{sec:recap}

We use the algorithmic plan for Pauli error estimation from \cite{FO21paulierror}. Following their exposition, two characters are introduced: Alice the learner, and Charlie the ``channel operator''. In particular, we think of Charlie as implementing the Pauli channel with error rates~$\pi$ according to the following process: When Alice wants to pass a state $\ket{\psi_A}$ through the channel, she sends it to Charlie, who picks a random~$\bC \sim \pi$, and then returns $\sigma_{\bC} \ket{\psi_A}$ to Alice.
In fact, our algorithm will only have Alice passing random $n$-qubit states composed of $\ket{+}, \ket{i}, \ket{0}$ through the channel.  Thus the setup can be summarized as below:
\begin{definition} \label{def:probe}
    To learn a Pauli channel defined by error rates~$\pi$, Alice repeatedly performs ``\textit{random nontrivial probes}'' as follows:
    \begin{enumerate}
    \item Alice picks a uniformly random string $\bA \in \{1, 2, 3\}^n$.
    \item Alice prepares an unentangled $n$-qubit state $\ket{\psi_{\bA}}$, where the $j$th qubit is $\ket{\chi_+^{{\bA}_j}}$.
    \item Alice sends $\ket{\psi_{\bA}}$ to Charlie. Charlie draws a string $\bC \in \{0, 1, 2, 3\}^n$ according to~$\pi$, and sends back $\sigma_{\bC}\ket{\psi_{\bA}}$.
    \item Alice performs an unentangled measurement on $\sigma_{\bC}\ket{\psi_{\bA}}$, where the $j$th qubit is measured in the $\ket{\chi_{\pm}^{{\bA}_j}}$ basis. She obtains a \emph{readout} $\bR \in \{+, -\}^n$, which we can relabel to be from $\{\textsf{0}, \textsf{1}\}^n$.
\end{enumerate}
\end{definition}

In the description of \Cref{def:probe}, Alice picks her random string $\bA$ first, and then Charlie (independently) draws his random string~$\bC$ second. However, it is equivalent and more helpful to imagine Charlie ``secretly'' drawing $\bC$ first, and then Alice picking~$\bA$ next (without knowledge of~$\bC$); thus, we think of the probe based on $\bA$ as acquiring information about Charlie's outcome~$\bC = C$.  
With this interpretation in mind, it is not hard to see the  following fact~\cite[Fact~17]{FO21paulierror}: 
\begin{fact} \label{fact:13}
Fix an outcome $C \in \{0, 1,2, 3\}^n$ for Charlie's draw.  Now when Alice performs a random nontrivial probe, the coordinates of her readout $\bR$ are independent, and for each $1 \leq j \leq n$: 
\begin{enumerate}
    \item If $C_j = 0$, then $\bR_j = \textsf{\emph{0}}$ with probability~$1$.
    \item If $C_j \neq 0$, then with probability $\frac{1}{3}$ we have $\bR_j = \textsf{\emph{0}}$, and with probability $\frac{2}{3}$ we have $\bR_j = \textsf{\emph{1}}$.
\end{enumerate}
\end{fact}
\noindent (Briefly, this is because if Charlie will apply the identity Pauli $\sigma_0$ to a qubit, it will never flip from a $+1$ eigenvector to a $-1$ eigenvector; but, if Charlie will apply a non-identity Pauli~$\sigma$ to a qubit, there is a $\frac23$ chance that Alice chose its state to be a $-1$ eigenvector for~$\sigma$.)

Another way to say \Cref{fact:13} is that given $C$, the readout~$\bR$ is distributed as $\Xi(C)$, where $\Xi$ is defined as follows:
\begin{definition} \label{def:xi}
    We define the classical channel
    \begin{equation}
        \Xi : \{0,1,2,3\} \to \{\textsf{0},\textsf{1}\}, \qquad 0 \mapsto \textsf{0}, \qquad 1,2,3 \mapsto \begin{cases}
            \textsf{0} & \text{with probability $\frac13$}, \\
            \textsf{1} & \text{with probability $\frac23$.}
        \end{cases}
    \end{equation}    
\end{definition}
\noindent Thus Alice's strategy of random nontrivial probes transforms the Pauli channel estimation task into the Population Recovery problem with channel~$\Xi$.

We next observe another fact, \cite[Observation~16]{FO21paulierror}:
\begin{fact}
    Fix any $B \in \{0,1,2,3\}^n$.  Let $C \in \{0,1,2,3\}^n$ be the outcome of Charlie's draw.  
    If Alice does a random nontrivial probe with $\bA \in \{1,2,3\}^n$, and then she flips each bit $\bR_j$ for which $\sigma_{\bA_j}$ and $\sigma_{B_j}$ anticommute, then $\bR$ is distributed as $\Xi(B \oplus C)$.
\end{fact}
This means that for any $B \in \{0,1,2,3\}^n$ of Alice's choice, she can simulate the Population Recovery problem (with channel~$\Xi$) under the ``$B$-altered distribution'' $\pi^{\oplus B}$, defined by $\pi^{\oplus B}(C) = \pi(B \oplus C)$.
Finally, suppose that in this scenario she can efficiently estimate $\pi^{\oplus B}(0^n) = \pi(B)$ for $B$'s of her choice.  This means she can solve the Individual Recovery problem for $\pi$ under~$\Xi$, hence the general Population Recovery problem (by \Cref{prop:individ}).
\medskip

In summary, \cite{FO21paulierror} show that Pauli error estimation efficiently reduces to the following Individual Recovery task:  Given samples from $\Xi \circ \pi$, estimate $\pi(0^n)$ to additive precision~$\epsilon$.  Finally, it is easy to observe that in this scenario, there is no harm in merging the Pauli symbols $1,2,3$ into the bit $\textsf{1}$; this does not change $\pi(0^n)$, and it converts the channel~$\Xi$ into Z${}_\frac13$, the binary Z-channel with crossover probability~$\frac13$.
This completes the recap of how \cite{FO21paulierror} reduces (SPAM-free) Pauli error estimation to Individual Recovery of $\pi(\textsf{0}^n)$ under noise channel Z${}_\frac13$.

\section{Modeling SPAM} \label{sec:spam}
Operationally, SPAM errors might occur in a variety of natural ways.  To take state preparation errors for example, when an algorithm tries to prepare a certain state $\ket{a}$, several things might happen:  $\ket{a}$ might be replaced by a random pure state at a fixed angle from~$\ket{a}$; or, $\ket{a}$ might be replaced with a random mixed state $\boldsymbol{\rho}$ having a fixed fidelity with $\ket{a}$; or, $\ket{a}$ might be replaced with the maximally mixed state with a certain probability.
Mathematically, these are all equivalent to $\ket{a}$ being passed through a depolarizing channel.  We therefore make the following definition:
\begin{definition}  \label{def:sp}
    We model single-qubit state preparation error, with ``retention parameter'' $\rprep \in [0,1]$, by assuming that when $\ket{a}$ is intended, the actual produced state is $\rprep \cdot \ket{a} \hspace{-.3em}\bra{a} + (1 - \rprep)\cdot \frac12 I$, where $\frac12 I$ corresponds to the maximally mixed single-qubit state. Equivalently, the produced state has expected fidelity $\frac12 + \frac12 \rprep$ with the intended state.
\end{definition}

\begin{figure}[ht]
\centering
\begin{tikzpicture}[scale=1.1]
  \draw[cyan, thick, fill=cyan!15] (0,0) circle (0.9);
  \draw[cyan, dashed] (0,0) ellipse (0.9 and 0.3);
  
  \draw[gray, thick] (0,0) circle (1.5);
  \draw[gray, dashed] (0,0) ellipse (1.5 and 0.5);
  
  \draw[->] (0,-1.7) -- (0,1.9) node[above] {$|0\rangle$};
  \draw[->] (-1.8,0) -- (1.8,0) node[right] {$|1\rangle$};
  
  \fill[blue] (0,1.5) circle (2pt);
  \node[blue, right, font=\small] at (0.05,1.65) {$|a\rangle$};
  
  \fill[red] (0,0.9) circle (2pt);
  \node[red, right, font=\small] at (0.1,1.05) {$r_{\mathrm{prep}} \cdot \vec{a}$};
  
  \draw[-{Stealth}, thick, purple!70] (0,1.4) -- (0,1.0);
  
  \fill[gray] (0,0) circle (1.5pt);
  \node[gray, below right, font=\footnotesize] at (0.0,-0.0) {$\tfrac{1}{2}I$};
  
  \node[below, font=\small] at (0,-2.0) {Bloch vector contracts: $\vec{a} \mapsto r_{\mathrm{prep}} \cdot \vec{a}$};
\end{tikzpicture}
\caption{The depolarizing SPAM model on the Bloch sphere. All isotropic noise models (replacement with a random pure state at a fixed angle, replacement with a mixed state of fixed fidelity, or mixing with the maximally mixed state) are equivalent to contracting the Bloch vector by the retention parameter~$\rprep$, mapping pure states on the unit sphere to the interior of a smaller sphere of radius~$\rprep$.}
\label{fig:bloch_spam}
\end{figure}
\begin{remark}\label{rem:spam_assumptions}
    Our model assumes for simplicity that we have the same retention parameter for every intended state $\ket{a}$.  We also assume that the learning algorithm knows the parameter~$\rprep$ from prior calibration. The same remarks hold for \Cref{def:am} below.
\end{remark}

\begin{remark}\label{rem:spam_scope}
We emphasize that the equivalence between the various operational descriptions of SPAM listed above and the depolarizing model (\Cref{def:sp,def:am}) relies on an isotropy assumption: the noise is assumed to act symmetrically over all directions on the Bloch sphere (e.g.~the random perturbation of $\ket{a}$ is uniformly distributed over states at a given angle from $\ket{a}$). Accordingly, the robustness guarantee of our algorithm (\Cref{thm:mainresult}) is with respect to SPAM noise that is well-described by the depolarizing model, rather than arbitrary SPAM noise in full generality.

We also note that our model does not explicitly account for \emph{erasure errors}, which can be a significant source of error in some platforms such as neutral-atom and photonic systems. However, if the erasure event is heralded (i.e.~the learner knows when a qubit is lost), then the erased measurement outcomes can simply be discarded, reducing to our existing framework with a decrease in the effective sample size. Unheralded erasure errors, on the other hand, would require a different noise model and are not covered by our analysis.
\end{remark}
Similarly, for modeling measurement errors, there are several natural operational possibilities: when measuring $\rho$ against an intended outcome $\ket{a}$, one might actually get the result of measurement against some other random outcome $\ket{\boldsymbol{b}}$ making a certain angle from~$\ket{a}$; or, $\rho$ might get perturbed to a random state with a certain expected fidelity before the correct measurement is made, etc.  Once again, these possibilities are mathematically equivalent to $\rho$ being passed through the depolarizing channel.

\rev{To see this equivalence, note that any noise model in which the state $\rho$ is replaced by a random state $\boldsymbol{\rho}'$ whose expected Bloch vector satisfies $\E[\vec{\boldsymbol{\rho}'}] = \rmeas \cdot \vec{\rho}$ has the same outcome distribution as the depolarizing channel $\rho \mapsto \rmeas \cdot \rho + (1 - \rmeas) \cdot \frac{1}{2}I$. This is because measurement outcome probabilities depend linearly on the Bloch vector: for a measurement in basis $\{\ket{a}, \ket{a^\bot}\}$, the probability of outcome $\ket{a}$ is $\frac{1}{2}(1 + \vec{a} \cdot \vec{\rho'})$, and taking expectations gives $\frac{1}{2}(1 + \rmeas \cdot \vec{a} \cdot \vec{\rho})$, which matches the depolarized state.}

We therefore define:
\begin{definition} \label{def:am}
    We model single-qubit measurement error, with ``retention parameter'' $\rmeas \in [0,1]$, by assuming that when $\rho$ is measured in an intended orthonormal basis $\{\ket{a}, \ket{a^\bot}\}$, the outcome is as if the state $\rmeas \cdot \rho + (1 - \rmeas) \cdot \frac12 I$ were measured instead.
\end{definition}

\paragraph{The algorithm under SPAM.} We may now consider how the combined SPAM error affects our algorithm.
As described in \Cref{sec:recap}, when learner Alice makes random nontrivial probes to the Pauli channel defined by error rates~$\pi$, we may imagine that Charlie first draws $\bC \sim \pi$, obtaining some outcome $C \in \{0, 1, 2, 3\}^n$.  Given~$C$, everything else occurs in an unentangled fashion across qubits: For each $j \in [n]$, Alice prepares a random $\ket{\chi_+^{\bA_j}}$, state preparation error replaces this by $\rprep \cdot \ket{\chi_+^{\bA_j}} \hspace{-.3em}\bra{\chi_+^{\bA_j}} + (1 - \rprep)\cdot\frac{1}{2}I$, the channel changes this to $\rprep \cdot \sigma_{C_j} \ket{\chi_+^{\bA_j}}\hspace{-.3em}\bra{\chi_+^{\bA_j}} \sigma_{C_j}^\dagger + (1 - \rprep)\cdot\frac{1}{2}I$, and just before Alice's measurement in the $\ket{\chi_{\pm}^{\bA_j}}$ basis, measurement error changes the state to $\rprep \rmeas \cdot \sigma_{C_j} \ket{\chi_+^{\bA_j}}\hspace{-.3em}\bra{\chi_+^{\bA_j}} \sigma_{C_j}^\dagger + (1 - \rprep\rmeas) \cdot \frac12 I$.
This can be viewed as the SPAM-free result --- namely, $\sigma_{C_j}$ applied to $\ket{\chi_+^{\bA_j}}$ --- passed through the depolarizing channel with retention parameter $r = \rprep \cdot  \rmeas$.  In other words:
\begin{quotation}
    Given $C$, for each $j \in [n]$ independently, Alice gets the random readout she normally would (namely $\Xi(C)$) with probability~$r = \rprep \cdot \rmeas$, and a uniformly random bit from $\{\textsf{0}, \textsf{1}\}$ with probability $\delta \coloneqq 1-r$.
    In other words, she receives $\textrm{BSC}_{\frac{\delta}{2}} \circ \Xi(C)$.
\end{quotation}
It is easy to show that the reductions described after \Cref{def:xi} in \Cref{sec:recap} continue to hold.  Thus we finally obtain:
\begin{theorem} \label{thm:combo}
    The task of estimating the  error rates~$\pi$ of an $n$-qubit Pauli channel with SPAM governed by retention rates $\rprep, \rmeas$ reduces (with a factor $O(n \log (n/\eps))$ loss of efficiency) to the task of Individual Recovery of $\pi(\textsf{\emph{0}}^n)$ under noise channel $\textrm{\emph{BSC}}_{\frac{\delta}{2}} \circ \textrm{\emph{Z}}_{\frac13}$, where $\delta = 1 - \rprep \cdot \rmeas$.
\end{theorem}

\begin{remark} \label{rem:rc}
    As in \cite{FO21paulierror}, the above Theorem and our main \Cref{thm:mainresult} extend to the case of learning the Pauli error rates of a \emph{general} $n$-qubit channel~$\calE$. These are defined as the error rates of the Pauli channel formed from $\calE$ by Pauli twirling, namely $\rho \mapsto \E_{\bT \sim \{0,1,2,3\}^n} \sigma_{\bT}^\dagger \calE(\sigma_{\bT} \rho \sigma_{\bT}^\dagger) \sigma_{\bT}$. The details are exactly as in \cite[Section~6.1]{FO21paulierror}, with the only change to the statement of \Cref{thm:mainresult} being that Alice now randomly prepares qubits in one of the six pure states $\ket{\chi^j_{\pm}}$, rather than just the three $\ket{\chi^j_{+}}$.
\end{remark}

\section{Individual recovery for combined noise channels}
In the preceding sections, we reduced the Pauli error estimation problem (with SPAM) to an Individual Recovery task under the combined noise channel $\textrm{BSC}_{\frac{\delta}{2}} \circ \textrm{Z}_{\frac13}$ (\Cref{thm:combo}).

In this section, we derive the generating functions for this combined channel and establish the key quantity $\eta(\eps)$ that controls the sample complexity of Individual Recovery. The main result of this section, \Cref{thm:eta_to_complex_circles}, expresses $\eta(\eps)$ as an optimization problem involving polynomials evaluated on a complex circle, which we will then lower-bound in \Cref{sec:lowerbounding}.

The Individual Recovery problem is well understood separately for BSC and BEC~\cite{polyanskiy2017sample,DOS17populationrecovery}.  As noted in \cite{FO21paulierror}, the Z-channel behaves similarly to BEC in the context of Population Recovery.  However, handling the combined channel $\textrm{BSC}_{\frac{\delta}{2}} \circ \textrm{Z}_{\frac13}$ from \Cref{thm:combo} will require more work.  We will follow the methodology from~\cite{DOS17populationrecovery}.

\subsection{Generating functions}
Recall our goal is Individual Recovery of $\calD(\textsf{0}^n)$, given samples from an unknown probability distribution $\calD$ on $\{\textsf{0}, \textsf{1}\}^n$, masked by a binary noise channel~$\Lambda$.
As described in \cite{DOS17populationrecovery}, it is easy to assume without loss of generality that $\calD$ is symmetric, meaning it gives
equal probability mass to all strings with the same Hamming weight, where Hamming weight is defined as the number of $\textsf{1}$s in the string. This is because, given an arbitrary distribution $\mathcal{D}$, we can randomly permute each sample's coordinates, which is equivalent to sampling a symmetric distribution $\mathcal{D}^{\text{sym}}$, where $\mathcal{D}^{\text{sym}}(\textsf{0}^n) = \mathcal{D}(\textsf{0}^n)$. Therefore, it suffices to be able to estimate in the symmetric setting.

Given this, let us introduce the row vector 
$p = [p_0 \,\, p_1 \,\, \cdots \,\, p_n]$, where $p_i$ denotes the total probability mass $\calD$ has on Hamming weight~$i$.
Thus the learner's goal is to estimate~$p_0$.
Note also that the learner only needs to observe the Hamming weights of the samples it obtains (this is true even when $\Lambda$ is the BEC). Therefore, we can think of the learner as obtaining samples from $\{0, 1, \dots, n\}$ distributed according to the  probability row vector $q = [q_0 \,\, q_1 \,\, \cdots \,\, q_n]$, where 
\begin{equation}    \label{eqn:A}
    q = pA^{(\Lambda)}, \quad A^{(\Lambda)}_{ij} = \Pr[\text{a weight $i$ string becomes a weight $j$ string under $\Lambda$ noise}].
\end{equation}
The rows of each matrix $A^{(\Lambda)}$ can be nicely expressed using generating functions.  
For example, \cite[Propositions~2.1, 2.2]{DOS17populationrecovery}  are the following:
\begin{proposition}
    For the $A$ matrices associated with the $\textnormal{BEC}_\lambda$ and $\textnormal{BSC}_b$ channels, and with $z$ an arbitrary complex number:
    \begin{align}
        \sum_{j=0}^n    A^{(\textnormal{BEC}_\lambda)}_{ij} z^j &= (\lambda + (1-\lambda)z)^i \label{eqn:GBEC}\\
        \sum_{j=0}^n A^{(\textnormal{BSC}_b)}_{ij} z^j &= (b + (1-b)z)^i((1-b) + b z)^{n-i}  \label{eqn:GBSC}
    \end{align}
\end{proposition}

Also, because both the BEC${}_\lambda$ and the Z${}_\lambda$ channels convert $\textsf{1}$'s to non-$\textsf{1}$'s with the same probability,~$\lambda$, it is easy to see:
\begin{fact} \label{fact:Z}
    The Z-channel and BEC have the same $A$-matrices, $A^{(\textnormal{Z}_\lambda)} = A^{(\textnormal{BEC}_\lambda)}$.
\end{fact}

We would now like to derive the generating function for our combined channel,
$\textnormal{BSC}_{\frac{\delta}{2}} \circ \textnormal{Z}_{\frac13}$. Let us make a general definition:
\begin{definition} \label{def:zflip}
    For parameters $0 \leq \nu_1, \nu_2 < 1$, we define the \emph{ZFlip} channel $\textnormal{ZFlip}_{\nu_1, \nu_2} : \{\textsf{0},\textsf{1}\} \to  \{\textsf{0},\textsf{1}\}$ to be the concatenated channel $\textnormal{BSC}_{\frac{1-\nu_2}{2}} \circ \textnormal{Z}_{1-\nu_1}$. 
\end{definition}
The slightly strange parameterization makes formulas simpler, as $\nu_1, \nu_2$ may be thought of as ``retention rates''. For the Pauli error estimation problem with overall SPAM retention parameter $r = \rprep \cdot \rmeas$, we care about
\begin{equation}
    \nu_1 = \frac23, \qquad \nu_2 = r. 
\end{equation}


To compute the generating functions for the ZFlip channel, we start with the following simple fact:
\begin{proposition} \label{prop:concat}
    Given channels $\Lambda_1, \Lambda_2$, the ``$A$-matrix'' (as in \Cref{eqn:A}) for the concatenated channel satisfies $A^{(\Lambda_2 \circ \Lambda_1)} = A^{(\Lambda_1)} \cdot A^{(\Lambda_2)}$.
\end{proposition}
\begin{proof}
We have
\begin{align*}
    A^{(\Lambda_2 \circ \Lambda_1)}_{ij} &= \Pr[\text{a weight $i$ string becomes a weight $j$ string under $\Lambda_2 \circ \Lambda_1$ noise}] & \\
    &= \sum_{k = 0}^n \Pr[\text{weight $i$ $\to$ weight $k$ under $\Lambda_1$ noise}] \cdot \Pr[\text{weight $k$ $\to$ weight $j$ under $\Lambda_2$}] \\
    &= \sum_{k = 0}^n A^{(\Lambda_1)}_{ik} A^{(\Lambda_2)}_{kj} = (A^{(\Lambda_1)}  \cdot A^{(\Lambda_2)})_{ij}. \qedhere
\end{align*}
\end{proof}

Now we can derive the generating function for the $\text{ZFlip}_{\nu_1, \nu_2}$ channel:
\begin{proposition} \label{prop:ourG}
Given $0 \leq \nu_1, \nu_2 \leq 1$, write $b = \frac{1 - \nu_2}{2}$.  Then 
\begin{equation} \label{eqn:Gcombo}
    G_i(z) \coloneqq \sum_{j = 0}^n A^{(\textnormal{ZFlip}_{\nu_1,\nu_2})}_{ij} z^j = ((1-b)+bz)^n \cdot ((1-\nu_1) + \nu_1 w)^i,
    \qquad \text{where  $w \coloneqq \frac{b+(1-b)z}{(1-b)+bz}$.}
\end{equation}    
\end{proposition}
\begin{proof}
    For brevity, write $A^{(1)} = A^{(\textrm{Z}_{1-\nu_1})}$, which from \Cref{fact:Z} and \Cref{eqn:GBEC} we know has
    \begin{equation}
        \sum_{k=0}^n A^{(1)}_{ik}z^k = ((1-\nu_1) + \nu_1 z)^i.
    \end{equation}
    And for brevity, write $A^{(2)} = A^{(\textrm{BSC}_{b})}$, which from \Cref{eqn:GBSC} we know has
    \begin{align}
        \sum_{j=0}^n A^{(2)}_{kj} z^j = (b + (1-b)z)^k((1-b) + b z)^{n-k} 
        =((1-b) + b z)^{n} \cdot w^k.
    \end{align}
    Now \Cref{prop:concat} tells us that $\sum_{j = 0}^n A^{(\textnormal{ZFlip}_{\nu_1,\nu_2})}_{ij} z^j$ equals
\begin{align}
    \sum_{j = 0}^n \sum_{k = 0}^n A^{(1)}_{ik}A^{(2)}_{kj}z^j 
    &= \sum_{k = 0}^n A^{(1)}_{ik} \left( \sum_{j = 0}^n A^{(2)}_{kj}z^j \right)
    = ((1-b) + b z)^{n} \cdot \sum_{k=0}^n A^{(1)}_{ik} w^k \\
    &= ((1-b) + b z)^{n} \cdot ((1-\nu_1) + \nu_1 w)^i,
\end{align}
as claimed.
\end{proof}

Recall now that our task is to (with high probability) estimate $p_0$ to additive precision~$\eps$, given samples drawn from distribution~$q = p A^{(\Lambda)}$, where $\Lambda$ is the ZFlip channel of interest from \Cref{def:zflip}.  The work \cite{DOS17populationrecovery} precisely lower-bounds the sample complexity of this task, and gives an algorithmic upper bound matching to within polynomial factors:
\begin{theorem} \label{thm:dos}
    (Combination of \cite[Props.~3.1, 3.2, and subsequent, slight notation change]{DOS17populationrecovery}.)  For binary noise channel~$\Lambda$, define
    \begin{equation}
        \eta_\Lambda(\epsilon) = \displaystyle\min_{\substack{\textnormal{probability vectors }p, p' \\ \vert p_0 - p'_0\vert > \epsilon}} \ \lVert pA^{(\Lambda)} - p'A^{(\Lambda)} \rVert_1.
    \end{equation}
    Then the sample complexity of the Individual Recovery task is $\Omega(\frac{1}{\eta_\Lambda(2\epsilon)})/\sqrt{n}$.  Moreover, the task can be solved with $\poly(n, \frac{1}{\eta_\Lambda(\epsilon)})$ samples and running time.  
    
    Finally, if the generating function $\sum_{j=0}^n A^{(\Lambda)}_{ij} z^j$ is $G_i(z)$, then
    \begin{equation}
         \eta_\Lambda(\epsilon) \geq 
         \min_{\substack{c \in \Delta \\ c_0 > \epsilon}} \max_{\substack{z \in \C \\ \vert z \vert = 1}} \ \left\vert \sum_{j=0}^n c_j G_j(z)\right\vert,
    \end{equation}
    where $\Delta \coloneqq \{c = (c_0, c_1, \dots, c_n) : \sum_j c_j = 0, \ \sum_j |c_j| \leq 2\}$.
\end{theorem}

\begin{remark}\label{rem:reconcile}
We briefly explain how \Cref{thm:dos} yields the explicit sample complexity in \Cref{thm:mainresult}.
\begin{enumerate}
    \item \Cref{thm:dos} states that Individual Recovery can be solved with $\poly(n, 1/\eta_\Lambda(\eps))$ samples.
    \item \Cref{thm:etafinalbound} lower-bounds $\eta(\eps)$. In the high-SPAM regime, this gives \\ $1/\eta(\eps) \leq \exp(O((\delta n)^{1/3} \ln^{2/3}(1/\eps)))$. In the low-SPAM regime, it gives $1/\eta(\eps) \leq \eps^{-O(1)}$.
    \item Combining with the $O(n\log(n/\eps))$ overhead from the branch-and-prune reduction (\Cref{prop:individ} and \Cref{thm:combo}) gives the stated complexity in \Cref{thm:mainresult}. The polynomial factors in $n$ and $\eps$ are absorbed by the $\poly(\cdot)$ and $\exp(O(\cdot))$ notation.
\end{enumerate}
\end{remark}

In \Cref{thm:dos}, we need to understand $G_j(z)$'s values for $z$ on the complex unit circle.  Recall from \Cref{prop:ourG} that for our ZFlip channel of interest, with the M\"{o}bius transformation $w = \frac{b+(1-b)z}{(1-b)+bz}$, we have
\begin{equation}
    G_j(z) = ((1-b)+bz)^n \cdot ((1-\nu_1) + \nu_1 w)^j = \left(\frac{1-2b}{(1-b)-bw}\right)^n  \cdot ((1-\nu_1) + \nu_1 w)^j,
\end{equation}
where we solved for $z$ in terms of~$w$. It is easy to see that in the formula for~$w$, the numerator and denominator of $w$ have the same squared length when $|z| = 1$; i.e., we have $|w|^2 = 1$.  Since M\"{o}bius transformations map circles to circles, we conclude that the locus of $w$ for $|z| = 1$ is also the unit circle, $|w| = 1$.\footnote{Unless $b = \frac12$, in which case $w$ is constantly~$1$.  This corresponds to the trivial case of maximal SPAM, $\nu_2 = 0$, which we henceforth exclude.}
Now writing $w = e^{i\theta}$ for $\theta \in (-\pi, \pi]$, 
we have 
\begin{equation}
    G_j(\theta) =  \left(\frac{\nu_2}{(1-b)-be^{i\theta} }\right)^n  \cdot ((1-\nu_1) + \nu_1 e^{i\theta})^j,
\end{equation}
where recall $b = \frac{1-\nu_2}{2}$. Using the notation $Q(v) = \sum_{j=0} c_j v^j$,  \Cref{thm:dos} tells us that 
for $\Lambda = \textnormal{ZFlip}_{\nu_1,\nu_2}$,
\begin{equation}
    \eta_{\Lambda}(\eps) \geq \min_{Q} \max_{-\pi < \theta \leq \pi} \ \left\vert\frac{\nu_2}{(1-b)-be^{i\theta} }\right\vert ^n \cdot \left|Q\left((1-\nu_1) + \nu_1 e^{i\theta}\right)\right|,
\end{equation}
where the minimum is over real-coefficient polynomials of degree at most~$n$ satisfying $Q(0) > \eps$, $Q(1) = 0$, and $L(Q) \coloneqq \sum_j |c_j| \leq 2$.  Finally, note that
\begin{multline}
    \abs{(1-b)-be^{i\theta}}^2 = (1-b)^2 + b^2 - 2b(1-b) \cos\theta \\
    = \frac{1+\nu_2^2}{2} - \frac{1-\nu_2^2}{2} \cos \theta 
    = \frac{1+\nu_2^2}{2} - \frac{1-\nu_2^2}{2} (1-2\sin^2(\theta/2)) = \nu_2^2 + (1-\nu_2^2)\sin^2(\theta/2).
\end{multline}
Thus
\begin{equation}
    \left|\frac{\nu_2}{(1-b)-be^{i\theta} }\right|^n = \left(\frac{\nu_2^2}{\nu_2^2 + (1-\nu_2^2)\sin^2(\theta/2)}\right)^{n/2} = \left(1 + \frac{1-\nu_2^2}{\nu_2^2}\sin^2(\theta/2)\right)^{-n/2}.
\end{equation}
Putting everything together, we conclude:

\begin{theorem}\label{thm:eta_to_complex_circles}
    Individual Recovery of $\pi(\textsf{\emph{0}}^n)$ under the noise channel $\textnormal{ZFlip}_{\nu_1, \nu_2}$ can be accomplished with $\poly(n, 1/\eta(\eps))$ samples and running time, where the function $\eta(\eps)$ satisfies
    \begin{equation}
        \eta(\epsilon)\geq \min_{Q} \max_{-\pi < \theta \leq \pi} \left(1 + \frac{1-\nu_2^2}{\nu_2^2}\sin^2(\theta/2)\right)^{-n/2} \cdot \left\vert Q((1 - \nu_1) + \nu_1 e^{i\theta})\right\vert
    \end{equation}
    and the minimum is over real-coefficient polynomials $Q$ of degree at most $n$ satisfying $Q(0) > \epsilon$, $Q(1) = 0$, and $L(Q) \leq 2$.
\end{theorem}

\section{Lower-bounding $\eta(\epsilon)$} \label{sec:lowerbounding}
We now turn to the main technical contribution of the paper: lower-bounding $\eta(\eps)$ via \Cref{thm:eta_to_complex_circles}. The quantity we must control is the modulus of a polynomial $Q$ at the points $(1-\nu_1) + \nu_1 e^{i\theta}$, which trace out the circle of radius $\nu_1$ centered at $1 - \nu_1$. This circle passes through $1$ on the real axis, where $Q(1) = 0$. The difficulty is that we need $|Q|$ to be large on an \emph{arc} near this root, even though $Q$ vanishes at the root itself and its coefficients are constrained. Classical results of \cite{BE97} give such lower bounds, but only on the unit circle; our circle has smaller radius and a shifted center. The bulk of this section is devoted to extending their complex analysis to this setting.
 
The argument has three ingredients, developed in \Cref{subsec:lemmas}. The first is a Hadamard three-line lemma (\Cref{lem:threeline}) that lets us trade off bounds on nested arcs. The second is an arc lower bound for functions in a suitable analytic class on circles of radius slightly smaller than~$1$ (\Cref{cor:BE4.1}). The third, \Cref{thm:BE3.1}, combines these two to give the arc lower bound we need on $\partial D_{\nu_1}(1 - \nu_1)$. We apply the third in \Cref{subsec:mainproof} to prove the bound on $\eta(\eps)$.

\rev{\subsection{Statement of the main bound}}
We will prove the following bound:
\begin{theorem}\label{thm:etafinalbound}
\rev{For $0 < \epsilon < 1$ and $n \in \N$, with the ZFlip parameters $\nu_1 = \frac{2}{3}$ and $\nu_2 = 1-\delta$ from \Cref{def:zflip}, \rev{$0 \leq \delta \leq 0.99$,} and $\delta_c := \dfrac{6}{\pi^3}\cdot\dfrac{\ln(1/\epsilon)}{n}$, we have}
\begin{equation}
    \rev{\eta(\epsilon) \geq
    \begin{cases}
        \exp\left(-O\left((\delta n)^{1/3} \ln^{2/3}(1/\epsilon)\right)\right), & \text{if } \delta \geq \delta_c, \\[2.5ex]
        \epsilon^{O(1)}, & \text{if } \delta < \delta_c.
    \end{cases}}
\end{equation}
    where $\eta(\epsilon)$ characterizes the sample complexity and runtime of solving Individual Recovery of $\pi(\textsf{\emph{0}}^n)$ under the noise channel $\textnormal{ZFlip}_{\nu_1, \nu_2}$.
\end{theorem}

\paragraph{\rev{Setup for the proof.}}

Note that for $-\pi < \theta < \pi$, $\theta^2/16 \leq \sin^2(\theta/2) \leq \theta^2/4$. Using the fact that $e^{-x} \leq (1 + x)^{-1}$ for all $x \geq 0$, we have that 
\begin{equation}
    \left(1 + \frac{1-\nu_2^2}{\nu_2^2}\sin^2(\theta/2)\right)^{-1} \geq \exp\left(-\frac{1-\nu_2^2}{4\nu_2^2} \theta^2 \right),
\end{equation}
and therefore,
\begin{equation}\label{eq:eta_gaussian_bound}
    \eta(\epsilon)\geq \min_{Q} \max_{-\pi < \theta \leq \pi} \exp\left(-\frac{1-\nu_2^2}{8\nu_2^2} \theta^2n \right) \cdot \left\vert Q((1 - \nu_1) + \nu_1 e^{i\theta})\right\vert.
\end{equation} 

Next, fix an arbitrary vector $c = [c_0 \, c_1 \cdots c_n] \in \Delta$ with $c_0 > \epsilon$. Let $Q_c$ denote the polynomial $Q$ with coefficients $c$. \rev{We now define
\begin{equation}
\label{eq:Fc_def}
    F_c = \max_{-\pi < \theta \leq \pi} \exp\left(-\frac{1-\nu_2^2}{8\nu_2^2} \theta^2n \right) \cdot \left\vert Q_c((1 - \nu_1) + \nu_1 e^{i\theta})\right\vert.
\end{equation}}

\rev{Since $\eta(\epsilon) \geq \min_c F_c$, it suffices to lower-bound $F_c$ for our arbitrary $c$.} $F_c$ is controlled by $\vert Q_c \vert$ on the circle $\partial D_{\nu_1}(1 - \nu_1)$ of radius $\nu_1$ centered at $1 - \nu_1$, so we now need to establish a lower bound for the modulus of $Q_c$ on this circle. Actually, we will prove a lower bound for a slightly modified polynomial 
\begin{equation}
    \widetilde{Q}_c(z) = \frac{1}{c_0}Q_c(z) = \displaystyle\sum_{i=0}^n \Tilde{c}_i z^i,
\end{equation}
where $\Tilde{c_0} = 1$, and $\vert \Tilde{c_i} \vert \leq \frac{1}{c_0}$ for all $i$. In our proof, we will only need $\widetilde{Q}_c$ to satisfy the following two properties:

For $M \geq 1$ and $0 < \gamma \leq 1$,
\begin{equation}\label{eq:Qtilde_properties}
    \left| \widetilde{Q}_c\left(\frac{1}{4M}\right)\right| \geq \frac{1}{2}, \quad\text{and}\quad
    |z| \leq 1-\gamma \ \Rightarrow\  |\widetilde{Q}_c(z)| \leq M/\gamma. 
\end{equation}
\rev{These two properties play complementary roles in establishing a lower bound on $\widetilde{Q}_c$ over an arc of the circle near the point $z = 1$. The first guarantees that the modulus is non-negligible at an interior point of the disk. The second keeps the modulus from being too large on the subset of the circle boundary bounded away from $z = 1$. This, together with the maximum modulus, forces $|\widetilde{Q}_c|$ to be large on the arc near $z = 1$, which is the bound we ultimately require.}

The first property holds, since, if we set $M$ to $\frac{1}{c_0}$, and let $z_0 = \frac{1}{4M}$, then
\begin{equation}
    \vert \widetilde{Q}_c(z_0) \vert \geq 1 - \sum_{i = 1}^n M(z_0)^i = 1 - \frac{Mz_0(1 - z_0^n)}{1 - z_0} \geq 1 - \frac{M z_0}{1 - z_0} \geq \frac{1}{2},
\end{equation}
as desired. The second property holds, since, for $0 < \gamma \leq 1$ and for $z \in \C$ such that $\vert z \vert \leq 1 - \gamma$, we have that
\begin{equation}
    \vert \widetilde{Q}_c(z) \vert \leq \sum_{i = 0}^n \vert \Tilde{c_i} \vert \vert z \vert^i \leq \sum_{i = 0}^n \frac{1}{c_0}\vert z \vert^i \leq \sum_{i = 0}^n M \vert z \vert^i= \frac{M(1 - \vert z \vert^{n+1})}{1 - \vert z \vert} \leq \frac{M}{\gamma}.
\end{equation}

Since our proof only requires these two properties, we will prove our lower bound for the set of all polynomials that satisfy them, which we define in the following section.

\rev{\subsection{Function class and key lemmas}\label{subsec:lemmas}}
\rev{We now introduce the function class and the key lemmas needed for the proof of \Cref{thm:etafinalbound}. These are stated and proved independently, and are then combined in \Cref{subsec:mainproof} to complete the argument.}

\begin{definition}\label{def:C_class}
    Let $0 < \kappa \leq 1$ and $\lambda \geq 1$. $\calC_{\kappa, \lambda}$ is the set of all analytic functions $f$ that are continuous on the closed unit disk \rev{for which there exists $M \geq 1$ such that}
    \begin{equation}
        \left\vert f\left( \frac{1}{4M} \right) \right\vert \geq \kappa, \quad\text{and}\quad \rev{\text{for all } z \in \C \text{ with }} \vert z \vert \leq 1 - \gamma \ \Rightarrow\ \vert f(z) \vert \leq \frac{\lambda M}{1 - \vert z \vert}\rev{,}
    \end{equation}
\rev{for all $0 < \gamma \leq 1$.}
\end{definition}

\rev{\noindent (Recall from \Cref{eq:Qtilde_properties} that $\widetilde{Q}_c \in \calC_{1/2,1}$ with $M = 1/c_0 \geq 1/\eps$.)}
We will prove our lower bound for all polynomials in the set $\calC_{\frac12, 1}$, on an arc of the circle $\partial D_{\nu_1}(1 - \nu_1)$ with central angle $\theta \in (0, \pi)$. This arc has endpoints $\alpha$ and $\beta$, where  
\begin{equation}\label{eq:alpha_and_beta}
   \alpha = \nu_1\left[\cos(\theta/2) + i \sin(\theta/2) \right] + 1 - \nu_1,  \text{ and } \beta = \nu_1\left[\cos(\theta/2) - i \sin(\theta/2) \right] + 1 - \nu_1. 
\end{equation} 
Before delving into the proof, we must establish several lemmas. The first is derived from the Hadamard three-line theorem. We define 
\begin{equation}
    I_t = \left\{ z \in \C \mid \arg\left(\frac{\alpha - z}{z - \beta} \right) = t \right\},
\end{equation}
following the notation used in \cite{BE97}. Note that 
\begin{equation}
    \arg\left(\frac{\alpha - z}{z - \beta} \right) = \arg(\alpha - z) - \arg(z - \beta).
\end{equation}
$\arg(\alpha - z)$ represents the argument of the vector from point $z$ to point $\alpha$, while $\arg(z - \beta)$ represents the argument of the vector from point $\beta$ to point $z$. The difference between these two arguments is 0 when $z$ lies on the chord between $\alpha$ and $\beta$. Therefore, $I_0$ corresponds to this chord. Furthermore, by the inscribed angle theorem, the difference between these two arguments is $\frac{\theta}{2}$ when $z$ lies on the smaller arc of $\partial D_{\nu_1}(1 - \nu_1)$ with endpoints $\alpha$ and $\beta$, as shown below. Therefore,  $I_{\theta/2}$ corresponds to this arc.

\begin{figure}[H]
    \centering
    \includegraphics[width=8cm, height=6cm]{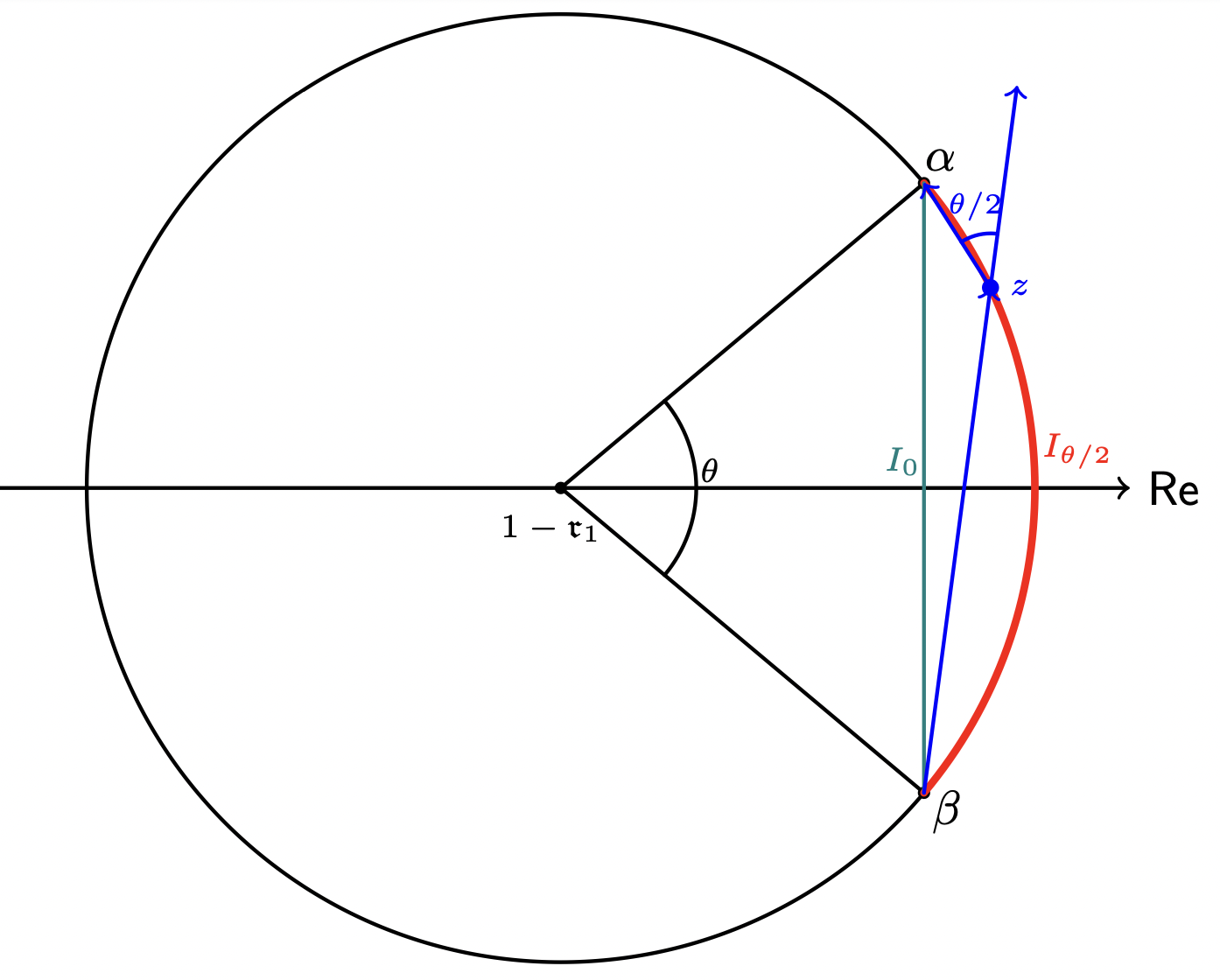} 
    \caption{Chord $I_0$ and arc $I_{\theta/2}$ on circle $\partial D_{\nu_1}(1 - \nu_1)$.}
    \label{fig:enter-label}
\end{figure}

We can again apply the inscribed angle theorem on the larger circle with center $1 - 2\nu_1$ and radius $2\nu_1\cos(\theta/4)$ to determine the arc that $I_{\theta/4}$ corresponds to, as shown below. 

\begin{figure}[H]
    \centering
    \includegraphics[width=8cm, height=8cm]{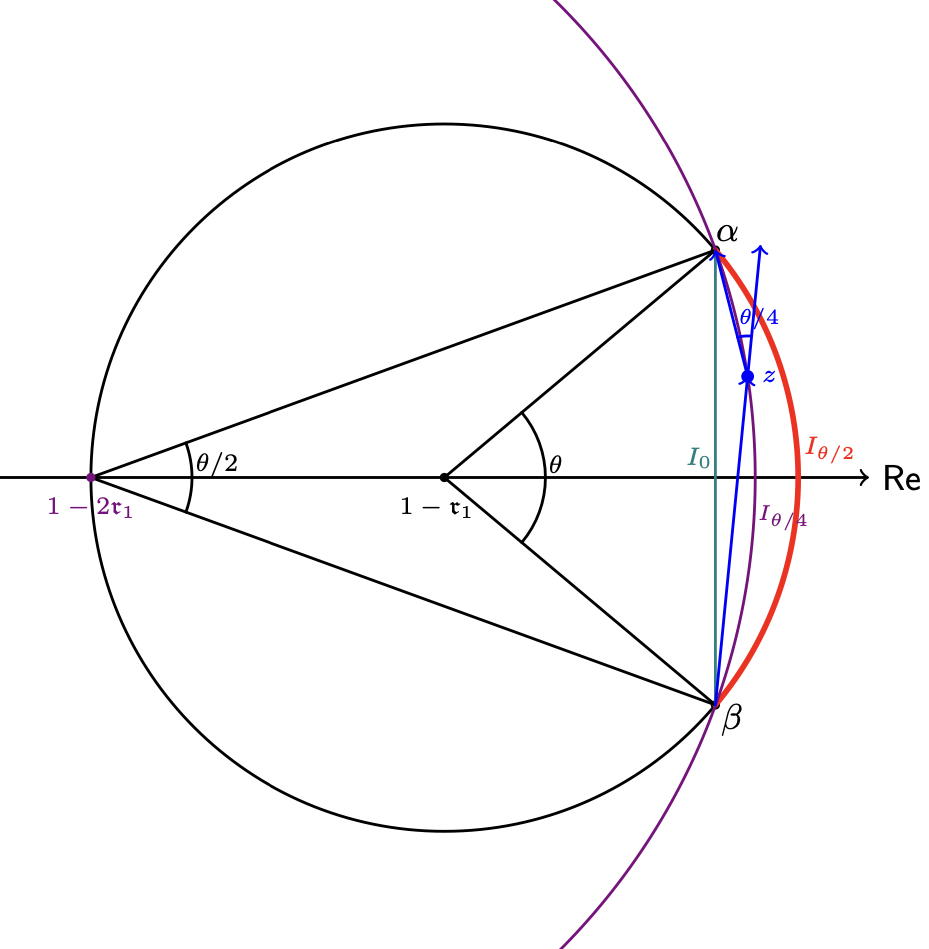} 
    \caption{Arc $I_{\theta/4}$ on circle $\partial D_{2\nu_1\cos(\theta/4)}(1 - 2\nu_1)$.}
    \label{fig:enter-label}
\end{figure}

Therefore, $I_{\theta/4}$ is the arc with endpoints $\alpha, \beta$ on the circle $\partial D_{2\nu_1\cos(\theta/4)}(1 - 2\nu_1)$. \\

In the following lemma, $I_0, I_{\theta/2}$, and $I_{\theta/4}$ are mapped to three infinite vertical lines on the complex plane. Then, the Hadamard three-line theorem is applied to upper-bound the maximum modulus of a function on $I_{\theta/4}$ using the maximum modulus of the same function on $I_0$ and $I_{\theta/2}$. This technique can be applied to any complex circle, not just $\partial D_{\nu_1}(1 - \nu_1)$. 

\begin{lemma}\label{lem:BE4.3}\label{lem:threeline} \textup{(Modified Lemma 4.3 in \cite{BE97}).} 
Consider the complex circle with radius $r$ and center $c$. Let $0 < \theta < \pi$,  $\alpha = r\left[\cos(\theta/2) + i \sin(\theta/2) \right] + c$, and $\beta = r\left[\cos(\theta/2) - i \sin(\theta/2) \right] + c$. Let $I_t = \left\{ z \in \C \mid \arg\left(\frac{\alpha - z}{z - \beta} \right) = t \right\}$. \\
Suppose $g$ is an analytic function \rev{in the open region bounded by $I_{0}$ and $I_{\theta/2}$}, and suppose $g$ is continuous on the closed region between $I_0$ and $I_{\theta/2}$. Then
\begin{equation}
   \max_{z \in I_{\theta/4}} \vert g(z) \vert \leq \left(\max_{z \in I_0} \vert g(z) \vert \right)^{1/2} \left(\max_{z \in I_{\theta/2}} \vert g(z) \vert \right)^{1/2}. 
\end{equation}
\end{lemma}

\begin{proof}
We can apply Lemma 4.2 in \cite{BE97}, which stems from the Hadamard three-line theorem:
\begin{lemma}\label{lem:BE4.2}
    \textup{(Lemma 4.2 in \cite{BE97}).} Let $t > 0$ and 
    \begin{equation}
        E_t := \{z \in \C \mid 0 \leq \operatorname{Im}(z) \leq t\}.
    \end{equation}
    \rev{Suppose $g$ is an analytic function in the interior of $E_{t}$, and suppose $g$ is continuous on $E_{t} \cup \infty$. Then} 
    \begin{equation}
        \rev{\max_{\{z \mid \operatorname{Im}(z) = t/2\}} \vert g(z) \vert \leq \left(\max_{\{z \mid \operatorname{Im}(z) = 0\}} \vert g(z) \vert \right)^{1/2} \left(\max_{\{z \mid \operatorname{Im}(z) = t\}} \vert g(z) \vert \right)^{1/2}.}
    \end{equation}

\end{lemma}
\noindent Let $w := \log\left(\frac{\alpha - z}{z - \beta}\right)$. Given $z \in I_{\theta/4}$, note that 
\begin{equation}
 \log\left(\frac{\alpha - z}{z - \beta}\right) = \log\left(\left\vert\frac{\alpha - z}{z - \beta}\right\vert\right) + i\arg\left(\frac{\alpha - z}{z - \beta}\right) = \log\left(\left\vert\frac{\alpha - z}{z - \beta}\right\vert\right) + \frac{i\theta
}{4} \in \{z \mid \operatorname{Im}(z) = \theta/4\}.   
\end{equation}
Similarly, given $z \in I_0$, $w \in \{z \mid \operatorname{Im}(z) = 0\}$, and given $z \in I_{\theta/2}$, $w \in \{z \mid \operatorname{Im}(z) = \theta/2\}.$
Now, we define a function $h$ such that $h(y) = g(w^{-1}(y))$. We thus have that 
\begin{align}
    \max_{z \in I_{\theta/4}} \vert g(z) \vert &= \max_{\{w \mid \operatorname{Im}(w) = \theta/4\}} \vert h(w) \vert \\
    &\leq \left(\max_{\{w \mid \operatorname{Im}(w) = 0\}} \vert h(w) \vert \right)^{1/2} \left(\max_{\{w \mid \operatorname{Im}(w) = \theta/2\}} \vert h(w) \vert \right)^{1/2} \tag{by \Cref{lem:BE4.2} \rev{with $t = \frac{\theta}{2}$}}\\
    &= \left(\max_{z \in I_0} \vert g(z) \vert \right)^{1/2} \left(\max_{z \in I_{\theta/2}} \vert g(z) \vert \right)^{1/2}. \qedhere
\end{align}
\end{proof}

The next theorem we will need is only a slight generalization of \cite[Lemma~4.1]{BE97}, but we prefer to give a full proof here, as the one in \cite{BE97} has a significant typo. 
\begin{theorem}\label{lem:BE4.1}
    Let $0 < \delta \leq \frac{1}{4}$, $0 < a < \pi$, $\lambda \geq 1$, and $M \geq 1$. Let $g$ be a continuous function on the unit disk, analytic on the interior, satisfying the following properties:
    \begin{equation}
        |g(\delta)| \geq 1, \quad\text{and}\quad
        |z| \leq 1-\gamma \ \Rightarrow\  |g(z)| \leq \frac{\lambda M}{\gamma}
    \end{equation}
    for $0 < \gamma \leq 1$. \\
    Let $\Gamma$ be the circle with center~$\delta$  and radius $1-\delta$. 
    Let $J$ be the closed arc of~$\Gamma$ with midpoint~$1$ and arc length~$a$.
    Then 
    \begin{equation}
       \max_{z \in J} |g(z)| \geq \left(\frac{\delta}{\lambda M}\right)^{O(1/a)}. 
    \end{equation}
\end{theorem}
\begin{proof}
\rev{The strategy is to construct a ``symmetrized'' function $h$ by evaluating $g$ at $2m$ equally spaced points on $\Gamma$. Since $g$ is large at the center $\delta$ and $h(\delta) = g(\delta)^{2m}$, the Maximum Modulus principle forces $h$ to be large somewhere on $\Gamma$. But the equally spaced points are arranged so that most of them miss the short arc $J$; at those distant points, $|g|$ is small by the bound $|g(z)| \leq \lambda M / \gamma$ for $|z| \leq 1-\gamma$. This forces the few factors of $h$ corresponding to points on $J$ to be large, yielding the desired lower bound.}

Let $2m$ be the smallest even integer greater than or equal to $\frac{4\pi}{a}$. We define $2m$ equally spaced points on $\Gamma$ such that the arc length between two adjacent points is approximately (but not larger than) $\frac{a}{2}$: 
\begin{equation}
  \eta_k := \delta + \left(1 - \delta\right)\xi^k, \qquad k \in \Z, -m \leq k \leq m,  
\end{equation}
where $\xi \coloneqq \exp(\frac{2\pi i}{2m})$ is the first $(2m)$-th root of unity. Note that the arc with endpoints $\eta_{-1}$ and $\eta_{1}$ that passes through $\eta_0$ is entirely contained within arc $J$. Also note that $\eta_{-m}$ and $\eta_{m}$ correspond to the same point. 

\rev{
\begin{figure}[ht]
\centering
\FigEquallySpacedPoints
\caption{Construction of equally spaced points $\eta_k$ on circle $\Gamma$ (center $\delta$, radius $1-\delta$) used in the proof of \Cref{lem:BE4.1}.}
\label{fig:equally_spaced_points}
\end{figure}
}

First, we will show that there exists a constant $c > 0$ such that for all $k \in \{1, ..., m - 1\}$, 
\begin{equation}\label{eq:mod_z_final_bound}
    \vert z \vert \leq 1 - c\delta(ka)^2 
\end{equation}
when $z$ is on the arc with endpoints $\eta_k$ and $\eta_{k +1}$ or $\eta_{- k}$ and $\eta_{-(k+1)}$. \\
This is because WLOG, assume that $z$ is on the arc with endpoints $\eta_k$ and $\eta_{k + 1}$. Then $z$ is of the form 
\begin{equation}
    \delta + \left(1 - \delta\right)e^{i\theta}, \quad \theta \in \left[\frac{\pi k}{m}, \frac{\pi(k + 1)}{m}\right].
\end{equation}
This means that 
\begin{equation}
    \vert z \vert^2 = 1 - 2\delta\left(1 - \delta\right)(1 - \cos\theta).
\end{equation}
Since $\vert z \vert \leq 1$, 
\begin{equation}\label{eq:mod_z_bound}
    \vert z \vert \leq \frac{1 + \vert z \vert^2}{2} = 1 - \delta\left(1 - \delta\right)(1 - \cos\theta) \leq 1 - \frac{3}{4}\delta(1 - \cos\theta).
\end{equation}
Since $\theta \in [0, \pi]$, 
\begin{equation}\
   1 - \cos\theta \geq \frac{2}{\pi^2}\theta^2 \geq \frac{2k^2}{m^2}. 
\end{equation}
Since $2m \leq c_1\frac{4\pi}{a}$ for some $c_1 > 0$:
\begin{equation}\label{eq:1-costheta_lower_bound}
     1 - \cos\theta \geq \frac{c_1}{2\pi^2} \cdot (ka)^2.
\end{equation}
Combining \Cref{eq:mod_z_bound} with \Cref{eq:1-costheta_lower_bound} gives us \Cref{eq:mod_z_final_bound}.

Now, we define the following function 
\begin{equation}
    h(z) := \prod_{j = -m}^{m-1} g\left(\delta + \xi^j\left(z - \delta \right)\right).  
\end{equation}

When $z \in \Gamma$, the $2m$ points plugged into $g$ to evaluate $h(z)$ are equally spaced on $\Gamma$. The arc length between two adjacent points is 
\begin{equation}
    \frac{2\pi}{2m}\left(1 - \delta \right) > \frac{a}{3}\left(1 - \delta \right) \geq \frac{a}{4} \tag{since $2m < \frac{4\pi}{a} + 2 \leq \frac{6\pi}{a}$}.     
\end{equation}
Therefore, the arc length between two adjacent points is strictly greater than a fourth of the arc length of $J$. This means that at most four of these points lie on arc $J$. 

Consider the remaining $2m - 4$ points. Each of these points lie on a distinct arc with endpoints $\eta_j$ and $\eta_{j+1}$, or endpoints $\eta_{-j}$ and $\eta_{-(j+1)}$, for $j = \{2, ..., m -1\}$. By \Cref{eq:mod_z_final_bound} and the properties of $g$, any point $z$ on the arc with endpoints $\eta_j$ and $\eta_{j+1}$, or endpoints $\eta_{-j}$ and $\eta_{-(j+1)}$, satisfies
\begin{equation}
    \vert g(z) \vert \leq \frac{\lambda M}{c\delta(ja)^2}. 
\end{equation}
Putting this all together,
\begin{align}
    \max_{z \in \Gamma} \vert h(z) \vert &\leq \left(\max_{z \in J} \vert g(z) \vert\right)^4 \cdot \prod_{k = 2}^{m -1} \left( \frac{\lambda M}{c\delta(ka)^2} \right)^2 \nonumber \\
    &= \left(\max_{z \in J} \vert g(z) \vert\right)^4 \left(\frac{\lambda M}{\delta}\right)^{2(m-2)} \left(\frac{1}{c_1a}\right)^{4(m-2)} ((m - 1)!)^{-4} \nonumber \\
    &\leq \left(\max_{z \in J} \vert g(z) \vert\right)^4 \left( \frac{\lambda M}{\delta} \right)^{2(m - 2)} \left( \frac{m}{c_1\pi} \right)^{4(m - 2)} \left(\frac{e}{m - 1}\right)^{4(m-1)} \tag{by Stirling's approximation, and since $a \geq \frac{2\pi}{m}$} \\
    &\leq \left(\max_{z \in J} \vert g(z) \vert\right)^4 \left(\frac{\lambda Me}{c_1\delta\pi} \right)^{2(m-2)} \left(\frac{m}{m - 1}\right)^{4(m-1)} e^{2m} \nonumber \\
    &\leq \left(\max_{z \in J} \vert g(z) \vert\right)^4 \left(\frac{\lambda Me}{c_1\delta \pi} \right)^{2(m-2)} e^{2m + 4} \tag{since $\left(1 + \frac{1}{n}\right)^n \to e$ as $n \to \infty$} \\
    &= \left(\max_{z \in J} \vert g(z) \vert\right)^4 \, \left(\frac{\lambda M}{\delta}\right)^{O(1/a)}.
\end{align}

Now, by the Maximum Modulus principle \rev{(\Cref{fact:mmp})}, 
\begin{equation}
   \left\vert g\left(\delta \right) \right\vert^{2m} = \left\vert h\left(\delta \right) \right\vert \leq \max_{z \in \Gamma} \vert h(z) \vert \leq  \left(\frac{\lambda M}{\delta}\right)^{O(1/a)} \left(\max_{z \in J} \vert g(z) \vert\right)^4.  
\end{equation}
So 
\begin{equation}
    \left(\max_{z \in J} \vert g(z) \vert\right)^4 \geq \left(\frac{\delta}{\lambda M}\right)^{O(1/a)} \left\vert g\left(\delta \right) \right\vert^{2m} \geq  \left(\frac{\delta}{\lambda M}\right)^{O(1/a)} \left(1\right) = \left(\frac{\delta}{\lambda M}\right)^{O(1/a)}. 
\end{equation}
Therefore 
\begin{equation*}
    \max_{z \in J} \vert g(z) \vert \geq \left(\frac{\delta}{\lambda M}\right)^{O(1/a)}. \qedhere    
\end{equation*}
\end{proof} 

Note that if we instead have that $g(\delta) \geq \kappa$ for some $0 < \kappa \leq 1$, then 
\begin{equation}
    \max_{z \in J} \vert g(z) \vert \geq \left(\frac{\delta\kappa}{\lambda M}\right)^{O(1/a)}.
\end{equation} 
This follows from replacing $g$ with $g/\kappa$ and $M$ with $M/\kappa$ in \Cref{lem:BE4.1}. \\

In our proof, we will use the following corollary, which provides a lower bound for $g$ on the arc of a circle with a slightly smaller radius that has a midpoint less than $1$.

\rev{
\begin{figure}[ht]
\centering
\FigGeneralizedLemma
\caption{The generalized setting of \Cref{cor:BE4.1}. The red arc $P$ lies on the circle $\Gamma_\mu$ of radius radius $1-\delta-\mu$, centered at $\delta$. The dashed arc on the outer circle $\Gamma$ has the same central angle as $P$.}
\label{fig:generalized_lemma}
\end{figure}
}
\begin{corollary}\label{cor:BE4.1}
    Let $0 < \delta \leq \frac{1}{4}$, $0 < a < \pi$, $\lambda \geq 1$, $M \geq 1$, and $0 < \kappa \leq 1$. Let $0 \leq \mu < 1 - \delta$. Let $g$ be a continuous function on the unit disk, analytic on the interior, satisfying the following properties:
    \begin{equation}
        |g(\delta)| \geq \kappa, \quad\text{and}\quad
        |z| \leq 1-\gamma \ \Rightarrow\  |g(z)| \leq \frac{\lambda M}{\gamma}, 
    \end{equation}
    where $\gamma \leq 1 - \delta$. \\
    Let $\Gamma_\mu$ be the circle with center~$\delta$  and radius $1-\delta - \mu$. 
    Let $P$ be the closed arc of~$\Gamma_\mu$ that is symmetric with respect to the real axis and has arc length~$a$.
    Then 
    \begin{equation}
        \max_{z \in P} |g(z)| \geq \left(\frac{\delta\kappa}{\lambda M}\right)^{O(1/a)}.
    \end{equation}
\end{corollary}
\begin{proof}
Consider the function 
\begin{equation}
    h(z) = g\left(\left(\frac{1 - \delta - \mu}{1 - \delta}\right)(z - \delta) + \delta\right). 
\end{equation}
First, note that $h(\delta) = g(\delta)$. Since $\vert g(\delta) \vert \geq \kappa$, we have that $\vert h(\delta) \vert \geq \kappa$. \\
Next, assume $\vert z \vert \leq 1 - \gamma$.
\begin{align*}
    \left\vert \left(\frac{1 - \delta - \mu}{1 - \delta}\right)(z - \delta) + \delta \right\vert &\leq  \left(\frac{1 - \delta - \mu}{1 - \delta}\right)\vert z \vert + \delta\left(1 - \frac{1 - \delta - \mu}{1 - \delta}\right) \tag{by triangle inequality} \\
    &\leq  \left(1 - \frac{\mu}{1 - \delta}\right)(1 - \gamma) + \delta\left(\frac{\mu}{1 - \delta}\right) \\
    &= 1 - \gamma - \left(\frac{\mu}{1 - \delta}\right)(1 - \gamma - \delta) \\
    &\leq 1 - \gamma \tag{since $1 - \gamma - \delta \geq 0$}.
\end{align*}
Therefore, 
\begin{equation}
    \vert h(z) \vert = \left\vert g\left(\left(\frac{1 - \delta - \mu}{1 - \delta}\right)(z - \delta) + \delta\right) \right\vert \leq \frac{\lambda M}{\gamma}.
\end{equation}
So $\vert z \vert \leq 1 - \gamma \Rightarrow \vert h(z) \vert \leq \frac{\lambda M}{\gamma}$. \\

As defined in \Cref{lem:BE4.1}, let $\Gamma$ be the circle with center $\delta$ and radius $1 - \delta$. When $z \in P$, $\left(\frac{1 - \delta - \mu}{1 - \delta}\right)(z - \delta) + \delta$ lies on the arc of $\Gamma$ with midpoint $1$ and the same central angle as $P$. Its arc length equals $\frac{1 - \delta}{1 - \delta - \mu}a \geq a$. Therefore, we can apply \Cref{lem:BE4.1} to get 
\begin{equation*}
   \max_{z \in P} \vert g(z) \vert = \max_{z \in J} \vert h(z) \vert \geq \left(\frac{\delta\kappa}{\lambda M}\right)^{O(1/a)}. \qedhere 
\end{equation*} 
\end{proof}

We will now prove our main result, which generalizes Theorem 3.1 in \cite{BE97} to complex circles beyond the unit circle.

\begin{theorem}\label{thm:BE3.1}
    \textup{(Generalized Theorem 3.1 in \cite{BE97}).} Let $0 < \theta < \pi$, $0 < \nu_1 \leq 1$, and $M \geq 1$. Let $A$ be an arc of $\partial D_{\nu_1}(1 - \nu_1)$ with central angle $\theta$ that is symmetric with respect to the real axis, passing through the point 1. Then 
    \begin{equation}
        \max_{z \in A} \vert f(z) \vert \geq \left( \frac{1}{M}\right)^{O\left(\frac{1}{\nu_1\theta}\right)}
    \end{equation}
    for all $f \in \calC_{\frac12, 1}$.
\end{theorem}

\begin{proof}
Note that the two endpoints of $A$ are $\alpha$ and $\beta$, as defined in \Cref{eq:alpha_and_beta}.

\rev{The proof proceeds in four steps (see \Cref{fig:proof_geometry_1,fig:proof_geometry_2} for a geometric overview):
\begin{enumerate}
    \item \textbf{Multiply $f$ by a factor that is zero at arc endpoints.} Define $g(z) = (z-\alpha)(z-\beta)f(z)$. The factor $(z-\alpha)(z-\beta)$ is zero at the arc endpoints and small near them, controlling the growth of $g$ near the boundary while ensuring it remains in a similar function class $\mathcal{C}$ with adjusted constants.
    \item \textbf{Bound $|g|$ on the chord $I_0$.} The chord connecting $\alpha$ and $\beta$ serves as the ``inner boundary.'' We bound $|g|$ on $I_0$ by combining the smallness of $|(z-\alpha)(z-\beta)|$ near the boundary with the property $|f(z)| \leq M/(1-|z|)$.
    \item \textbf{Interpolate via Hadamard three-line.} Apply \Cref{lem:BE4.3} to bound $|g|$ on $I_{\theta/4}$ in terms of $|g|$ on $I_0$ and on $A = I_{\theta/2}$, then use the Maximum Modulus principle on the region $G$ between $I_{\theta/4}$ and $A$.
    \item \textbf{Apply \Cref{cor:BE4.1} on a nearby arc.} Construct an arc $P$ passing through $G$ and obtain a lower bound on $|g|$ restricted to $P$, which transfers to $A$.
\end{enumerate}}

We define a new function $g$, where 
\begin{equation}
    g(z) = (z - \alpha)(z - \beta)f(z).
\end{equation} In this proof, we will derive a lower bound for the maximum modulus of $g$ on arc $A$, and then translate that to a lower bound for $f$. Note that since $f \in \calC_{\frac12, 1}$,
\begin{equation}
   \vert z \vert \leq 1 - \gamma \ \Rightarrow\ \vert g(z) \vert \leq \frac{M \vert (z - \alpha)(z - \beta) \vert}{\rev{1 - |z|}}  \leq \frac{4M}{\rev{1 - |z|}} \qquad \text{for } 0 < \gamma \leq 1, 
\end{equation} and 
\begin{equation}
    \textstyle \left\vert g\left(\frac{1}{4M}\right) \right\vert \geq \left(\frac{9}{32}\right)\left\vert f\left(\frac{1}{4M}\right) \right\vert \geq \left(\frac{9}{32}\right)\left(\frac{1}{2}\right) \geq \frac{1}{8}.
\end{equation}
Therefore, $g \in \calC_{\frac18, 4}$. 

\begin{figure}[ht]
\centering
\includegraphics[width=0.5\textwidth]{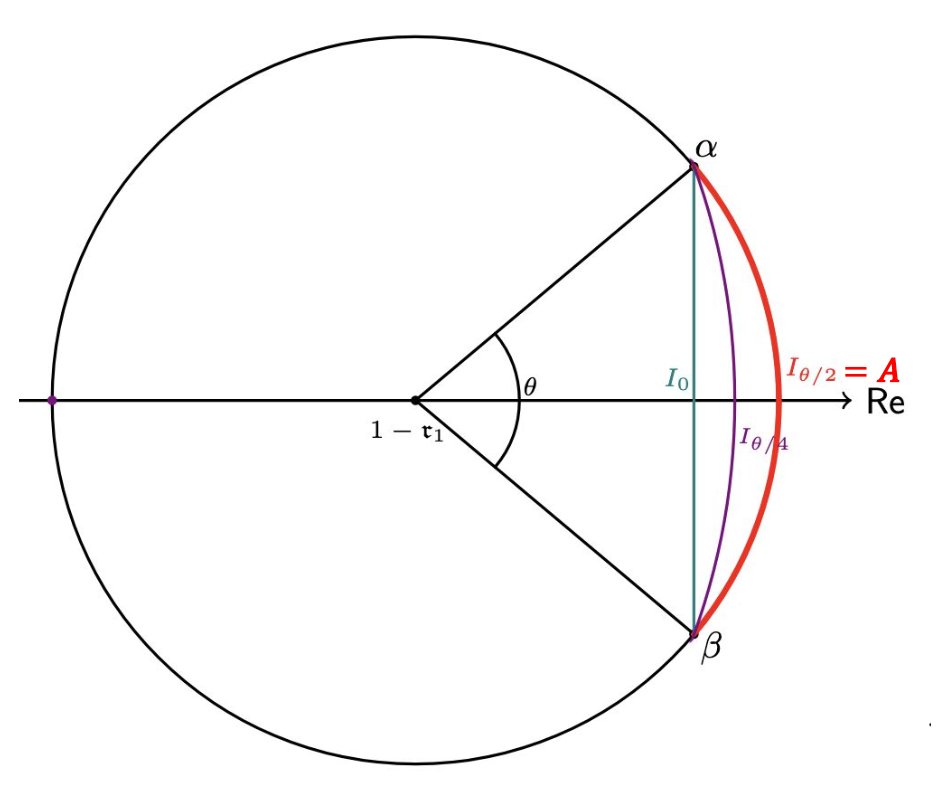}
\caption{Setup for the proof of \Cref{thm:BE3.1}. The red arc $A = I_{\theta/2}$ is the target arc on $\partial D_{\nu_1}(1-\nu_1)$. The teal chord $I_0$ connects endpoints $\alpha, \beta$. The violet arc $I_{\theta/4}$ lies on a larger circle through $\alpha, \beta$.}
\label{fig:proof_geometry_1}
\end{figure}

First, we wish to find a lower bound on $\vert g(z) \vert$ for $z \in I_0$. Recall that $I_0$ is the chord with endpoints $\alpha$ and $\beta$, as defined in \Cref{lem:BE4.3}. We split into two cases: when $0 < \nu_1 \leq \frac{3}{4}$, and when $\frac{3}{4} < \nu_1 \leq 1$. 
\paragraph{Case 1. $0 < \nu_1 \leq \frac{3}{4}.$} 
Note that on $I_0$, $\vert (z - \alpha)(z - \beta) \vert$ is maximized when $z$ is the midpoint of $I_0$, i.e. $z = \nu_1\cos(\theta/2) + 1 - \nu_1$. With this value of $z$, 
\begin{equation}
    \vert (z - \alpha)(z - \beta) \vert = \vert \left( -\nu_1 \sin(\theta/2)i \right) \left(\nu_1 \sin(\theta/2)i \right)\vert = \nu_1^2 \sin^2(\theta/2).
\end{equation}
Note that $1 - \vert z \vert$ is minimized when $z$ is either $\alpha$ or $\beta$. WLOG, assume $z = \alpha$. Then 
\begin{equation}
    1 - \vert z \vert = 1 - \sqrt{1 - 2\nu_1(1-\nu_1)(1-\cos(\theta/2))}.
\end{equation}
Now, the Taylor series of $1 - \sqrt{1 - x}$ is at least $\frac{x}{2}$ if $x$ is non-negative.
Since $2\nu_1(1 - \nu_1)(1 - \cos(\theta/2)) \geq 0$ when $0 \leq \theta \leq \pi$, we have that 
\begin{equation}
   1 - \vert z \vert \geq \nu_1(1 - \nu_1)(1 - \cos(\theta/2)). 
\end{equation}
So for $z \in I_0$, 
\begin{align}
    \vert g(z) \vert &\leq \frac{M\nu_1^2\sin^2(\theta/2)}{\nu_1(1 - \nu_1)(1 - \cos(\theta/2))} \nonumber \\
    &= \frac{M\nu_1\sin^2(\theta/2)}{(1 - \nu_1)(1 - \cos(\theta/2))} \nonumber \\
    &= \frac{M\nu_1\sin^2(\theta/2)}{(1-\nu_1)(2\sin^2(\theta/4))} \nonumber \\
    &= \frac{2M\nu_1\cos^2(\theta/4)}{1-\nu_1} \nonumber \\
    &\leq \frac{2M\nu_1}{1- \nu_1} \tag{since $\cos^2(\theta/4) \leq 1$ for $0 \leq \theta \leq \pi$} \\
    &\leq 8M\nu_1. \label{eq:I_0_case_1}
\end{align}

\paragraph{Case 2. $\frac{3}{4} < \nu_1 \leq 1.$} \, \\
WLOG, consider a point $z \in I_0$ whose distance from $\alpha$ is $\tau \in [0, \sin(\theta/2)]$. \\ 
Note that $z = \nu_1\cos(\theta/2) + 1 - \nu_1 + (\nu_1\sin(\theta/2) - \tau)i$. This means that
\begin{align}
    \vert z \vert &= \sqrt{(1 - \nu_1 + \nu_1\cos(\theta/2))^2 + (\nu_1\sin(\theta/2) - \tau)^2} \nonumber \\
    &= \sqrt{1 - 2\nu_1(1 - \nu_1)(1 - \cos(\theta/2)) - 2\tau\nu_1\sin(\theta/2) + \tau^2}.
    \intertext{Again, by Taylor series expansion, $1 - \sqrt{1 - x} \geq \frac{x}{2}$ for non-negative $x$. Therefore: }
    1 - \vert z \vert &\geq \nu_1(1 - \nu_1)(1 - \cos(\theta/2)) + \tau\nu_1\sin(\theta/2) - \frac{\tau^2}{2} \nonumber \\
    &\geq \tau\nu_1\sin(\theta/2) - \frac{\tau^2}{2} \nonumber \\
    &\geq \tau\nu_1\sin(\theta/2) - \frac{\tau\sin(\theta/2)}{2} \nonumber \\
    &\geq \frac{\tau\sin(\theta/2)}{4}.
\end{align}
Since $\vert (z - \alpha) (z - \beta) \vert = \tau \vert (z - \beta) \vert \leq 2\tau,$ we have that for $z \in I_0$,
\begin{equation}\label{eq:I_0_case_2}
    \vert g(z) \vert \leq \frac{8M}{\sin(\theta/2)}.
\end{equation}

Now, we can apply \Cref{lem:BE4.3} to \Cref{eq:I_0_case_1} and \Cref{eq:I_0_case_2}. Let $L = \displaystyle\max_{z \in I_{\theta/2}} \vert g(z) \vert = \max_{z \in A} \vert g(z) \vert$, since $I_{\theta/2}$ and $A$ correspond to the same arc. Then:
\begin{equation}
    \max_{z \in I_{\theta/4}} \vert g(z) \vert \leq \begin{cases}
\left(8M\nu_1L \right)^{1/2}, & \text{when } 0 < \nu_1 \leq \frac{3}{4} \\
\left(\frac{8ML}{\sin(\theta/2)} \right)^{1/2}, & \text{when } \frac{3}{4} < \nu_1 \leq 1
\end{cases},
\end{equation}

Next, we define $G$ to be the open region bounded by $I_{\theta/4}$ and $A$. By the Maximum Modulus principle \rev{(\Cref{fact:mmp})}, we have that 
\begin{equation}\label{eq:L_bound_max_modulus}
    \max_{z \in G} \vert g(z) \vert \leq \begin{cases}
\max \left\{L, \left( 8M\nu_1L\right)^{1/2} \right\}, & \text{when } 0 < \nu_1 \leq \frac{3}{4} \\
\max \left\{L, \left( \frac{8ML}{\sin(\theta/2)} \right)^{1/2} \right\}, & \text{when } \frac{3}{4} < \nu_1 \leq 1
\end{cases}.
\end{equation}

\begin{figure}[ht]
\centering
\includegraphics[width=0.5\textwidth]{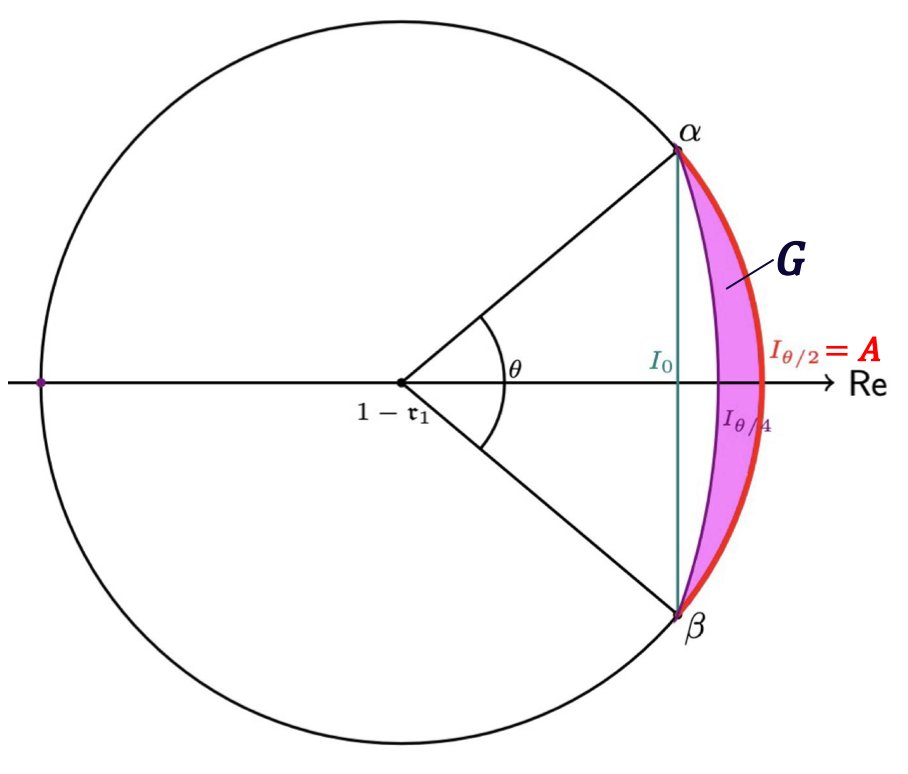}
\caption{The shaded region $G$ is bounded between $I_{\theta/4}$ and $A$. By the Maximum Modulus principle, $\max_{z \in G} |g(z)|$ is bounded by the maximum of $|g|$ on the boundary arcs $I_{\theta/4}$ and $A$.}
\label{fig:proof_geometry_2}
\end{figure}

Now, we want to define an arc that passes through the region $G$, so that we can apply \Cref{cor:BE4.1} and obtain a lower bound for the maximum modulus of $g$ on this arc. Note that $I_{\theta/4}$'s point of intersection with the real axis is
\begin{equation}
\textstyle
    1 - 2\nu_1 + 2\nu_1\cos(\theta/4) \leq 1 - 2\nu_1 + 2\nu_1\left(1 - \frac{\theta^2}{64}\right) = 1 - \frac{\nu_1\theta^2}{32}.
\end{equation}
Therefore, we define $\Gamma_P$ as the circle with center $\frac{1}{4M}$ and radius $1 - \frac{1}{4M} - \frac{\nu_1\theta^2}{64}$, as its point of intersection with the real axis is $1 - \frac{\nu_1\theta^2}{64} \in (1 - \frac{\nu_1\theta^2}{32}, 1)$. The choice of $64$ in the denominator is arbitrary; any constant exceeding $32$ is valid. We will now consider the arc $P := \Gamma_P \cap G$. 

Before applying \Cref{cor:BE4.1}, we first need to show that $P$ has arc length at least $\frac{\nu_1\theta}{c}$ for some constant $c \geq 1$. 

Let $a := \nu_1\theta$. Recall that the length of the chord $\vert I_0 \vert = 2\nu_1\sin(\theta/2)$. Therefore,
\begin{align*}
    \vert I_0 \vert &= \frac{2\sin(\theta/2)}{\theta}a \\
    &\geq \frac{2}{\pi} a &\tag{for $0 < \theta \leq \pi$}\\
    &\geq \frac{a}{c_1}. &\tag{for some $c_1 \geq 1$}
\end{align*}
So it suffices to show that $P$'s length is at least $\frac{|I_0|}{c_1}$ for some $c_1 \geq 1$.

Note that $I_{\theta/4}$ has radius $2\nu_1\cos(\theta/4) \leq 2\nu_1$. The radius of $P$ is 
\begin{equation}
\textstyle
   1 - \frac{1}{4M} - \frac{\nu_1\theta^2}{64} \geq \frac{3}{4} - \frac{\nu_1\theta^2}{64} \geq \frac{3}{4} - \frac{\nu_1}{4}. 
\end{equation}
We split into two cases: $\nu_1 \leq \frac{1}{3}$, and $\nu_1 > \frac{1}{3}$. 
\paragraph{Case 1. $\nu_1 \leq \frac{1}{3}$.} 
In this case, we have that $\frac{3}{4} - \frac{\nu_1}{4} \geq 2\nu_1$. This means that $I_{\theta/4}$ has a smaller radius (which corresponds to greater curvature) than $P$, and so the endpoints of arc $P$ are on arc $A$, and not on arc $I_{\theta/4}$. \\
\rev{Since the length of an arc is at least the length of the chord connecting its endpoints, it suffices to lower-bound $\vert P \vert$ by treating $P$ as a straight chord and computing its length.} Let $1 - t$ be the point at which this chord intersects the real axis ($t = \frac{\nu_1\theta^2}{64}$ based on how we've defined $P$). \\
Then, by the Pythagorean Theorem, $(\nu_1 - t)^2 + \left(\frac{1}{2}\vert P \vert \right)^2 = \nu_1^2$. This means that 
\begin{equation}
  \vert P \vert = 2\sqrt{t(2\nu_1 - t)} \geq 2 \sqrt{t(2\nu_1 - \nu_1)} = 2\sqrt{\nu_1t}.  
\end{equation}
Note that $I_0$ intersects with the real axis at the point $1 - \nu_1 + \nu_1\cos(\theta/2)$. Since $\cos(\theta/2) \geq 1 - \frac{(\theta/2)^2}{2} = 1 - \frac{\theta^2}{8}$, this means that the point of intersection with the real axis is $\geq 1 - \frac{\nu_1\theta^2}{8} = 1 - 8t$. \\
Therefore, 
\begin{equation}
    \vert I_0 \vert \leq 2\sqrt{t(16\nu_1 - t)} \leq 2\sqrt{16\nu_1t} = 4\vert P \vert.
\end{equation}

\paragraph{Case 2. $\nu_1 > \frac{1}{3}$.} 
In this case, $I_{\theta/4}$ has a larger radius than arc $P$, and so the endpoints of arc $P$ are on arc $I_{\theta/4}$ and not on arc $A$. \\
Note that $I_{\theta/4}$'s radius is $2\nu_1\cos(\theta/4) \leq 2$, and $P$'s radius is $1 - \frac{1}{4} - \frac{\nu_1\theta^2}{64} \geq \frac{1}{2}$. \rev{To obtain a lower bound on $\vert P \vert$, we substitute the radii that minimize it: radius $\frac{1}{2}$ for $P$ and $2$ for $I_{\theta/4}$.} Recall that with $t = \frac{\nu_1\theta^2}{64}$, $P$ intersects the real axis at point $1 - t$, and $I_{\theta/4}$ intersects the real axis at point $1 - 2t$. Therefore, $P$ intersects $I_{\theta/4}$ at the points 
\begin{equation}
    x = \frac{-3t^2-5t+3}{2t+3}, y = \pm \frac{\sqrt{-t(t-1)(t+3)(t+4)}}{2t+3}.
\end{equation}
Note that we can approximate $\frac{1}{2}\vert P \vert$ by calculating the length of the hypotenuse between point $1 - t$ and point $(x, +y)$. By the Pythagorean theorem, 
\begin{equation}
   \left(\frac{1}{2} \vert P \vert \right)^2 \geq (1 - t - x)^2 + y^2 = \frac{t(t+4)}{2t+3} \geq t. 
\end{equation}
Therefore, $\vert P \vert \geq 2\sqrt{t}$. Recall that 
\begin{equation}
   \vert I_0 \vert \leq 2\sqrt{16\nu_1t} \leq 2\sqrt{16t} = 4\vert P \vert.  
\end{equation}

Therefore, in both cases, we've established that the arc length of $P$ is at least $\frac{a}{c}$ for some $c \geq 1$. We can now apply \Cref{cor:BE4.1} with our function $g \in \calC_{\frac18, 4}$, and with $\delta = \frac{1}{4M}$ and $\mu = \frac{\nu_1\theta^2}{64}$. \rev{Note that the hypotheses of \Cref{cor:BE4.1} are satisfied since $\delta = \frac{1}{4M} \leq \frac14$ (as $M \geq 1$), the arc length of $P$ is at least $\frac{a}{c} > 0$, and $\mu = \frac{\nu_1\theta^2}{64} < 1 - \delta$ for $\theta \in (0, \pi)$}. This gives us: 
\begin{equation}\label{eq:result_from_4.1}
  \max_{z \in P} \vert g(z) \vert  \geq \left( \frac{1}{M}\right)^{O\left(\frac{1}{\nu_1\theta}\right)}.  
\end{equation}

We can now combine \Cref{eq:L_bound_max_modulus} and \Cref{eq:result_from_4.1} to derive our lower bound for $L = \displaystyle \max_{z \in A} \vert g(z) \vert$. When $\nu_1 \leq \frac{3}{4}$, we have that
\begin{equation}
   \max \left\{L, \left(8M\nu_1L \right)^{1/2} \right\} \geq \max_{z \in P} \vert g(z) \vert  \geq \left( \frac{1}{M}\right)^{O\left(\frac{1}{\nu_1\theta}\right)}. 
\end{equation}

If $L \geq  \left(8M\nu_1L \right)^{1/2}$, then we have that $L \geq \left( \frac{1}{M}\right)^{O\left(\frac{1}{\nu_1\theta}\right)}$. Otherwise, we have that 
\begin{equation}
    L \geq \left( \frac{1}{M}\right)^{O\left(\frac{1}{\nu_1\theta}\right)}\left(\frac{1}{\nu_1}\right) \geq \left( \frac{1}{M}\right)^{O\left(\frac{1}{\nu_1\theta}\right)}.
\end{equation}

When $\nu_1 > \frac{3}{4}$, we have that
\begin{equation}
    \max \left\{L, \left(\frac{8ML}{\sin(\theta/2)} \right)^{1/2} \right\} \geq \max_{z \in P} \vert g(z) \vert \geq \left( \frac{1}{M}\right)^{O\left(\frac{1}{\nu_1\theta}\right)}.
\end{equation}

If $L \geq \left(\frac{8ML}{\sin(\theta/2)} \right)^{1/2}$, then we have that $L \geq \left( \frac{1}{M}\right)^{O\left(\frac{1}{\nu_1\theta}\right)}$. Otherwise, we have that 
\begin{equation}
    L \geq \left( \frac{1}{M}\right)^{O\left(\frac{1}{\nu_1\theta}\right)} \sin(\theta/2).
\end{equation}
Note that for $0 < \theta < \pi$:
\begin{equation}
    \sin(\theta/2) \geq \left(\frac{1}{4}\right)^{\left(\frac{1}{\theta}\right)} \geq \left(\frac{1}{M}\right)^{O\left(\frac{1}{\nu_1\theta}\right)}. 
\end{equation}
Therefore, in both cases, we have that 
\begin{equation}
    L \geq  \left( \frac{1}{M}\right)^{O\left(\frac{1}{\nu_1\theta}\right)}.
\end{equation}
Putting everything together, we have that 
\begin{equation*}
    \max_{z \in A} \vert f(z) \vert \geq \frac{1}{4} \max_{z \in A} \vert g(z) \vert = \frac{L}{4} \geq  \left( \frac{1}{M}\right)^{O\left(\frac{1}{\nu_1\theta}\right)}. \qedhere 
\end{equation*}
\end{proof}

Note that the following corollary holds for \Cref{thm:BE3.1}:
\begin{corollary}\label{cor:BE3.2}
    \textup{(Generalized Corollary 3.2 of \cite{BE97}).} Let $Q(z)$ be a polynomial with complex coefficients, $\sum_{j = 0}^n b_j z^j$, such that $\vert b_0 \vert = 1$ and all coefficients $\vert b_j \vert \leq M$. Let $0 < \nu_1 \leq 1$, and let $A$ be an arc of $\partial D_{\nu_1}(1 - \nu_1)$ with central angle $0 < \theta < \pi$ that is symmetric with respect to the real axis and passes through the point $1$. Then there is some $w \in A$ such that
    \begin{equation}
        \vert Q(w) \vert \geq \left( \frac{1}{M}\right)^{O\left(\frac{1}{\nu_1\theta}\right)}.
    \end{equation}
\end{corollary}
\begin{proof}
  Polynomials of this form are part of set $\calC_{\frac12, 1}$, so \Cref{thm:BE3.1} can be applied. 
\end{proof}

\rev{\subsection{Completing the proof of \Cref{thm:etafinalbound}}\label{subsec:mainproof}}
\rev{We now combine the setup from \Cref{eq:eta_gaussian_bound}--\eqref{eq:Qtilde_properties} with the complex-analytic lemmas from \Cref{subsec:lemmas} to complete the proof.}

\begin{proof}[\rev{Proof of \Cref{thm:etafinalbound}}]
\rev{The idea is to balance two competing effects: the exponential decay factor $\exp(-\Theta(\theta^2 n))$ in \Cref{eq:eta_gaussian_bound} (which penalizes large $\theta$) against the polynomial lower bound from \Cref{cor:BE3.2} (which grows as $\theta$ increases). The optimal tradeoff occurs at the following choice of $\theta_0$:}

Let 
\begin{equation}
    \theta_0 = k \cdot \frac{\nu_2^{2/3}\ln^{1/3}(1/\epsilon)}{(\nu_1(1-\nu_2^2)n)^{1/3}}, \quad \text{for some constant $k \in (0, 1)$}.
\end{equation}
\rev{Recall from \Cref{eq:Fc_def} that}
\begin{align}
    \rev{F_c} &= \max_{-\pi < \theta \leq \pi} \exp\left(-\frac{1-\nu_2^2}{8\nu_2^2} \theta^2n \right) \cdot
    \left\vert Q_c(\nu_1 e^{i\theta} + 1 - \nu_1)\right\vert \nonumber \\
    &\geq \max_{-\theta_0 < \theta \leq \theta_0} \exp\left(-\frac{1-\nu_2^2}{8\nu_2^2} \theta^2n \right) \cdot
    \left\vert Q_c(\nu_1 e^{i\theta} + 1 - \nu_1)\right\vert \nonumber \\
    &\geq \exp\left(-\frac{1-\nu_2^2}{8\nu_2^2} \theta_0^2n \right) \cdot \max_{-\theta_0 < \theta \leq \theta_0}
    \left\vert Q_c(\nu_1 e^{i\theta} + 1 - \nu_1)\right\vert. 
\end{align}

\rev{We first establish the relationship between $\delta$ and $\theta_0$. Setting $\nu_1 = \frac{2}{3}$ and $\nu_2 = r = 1-\delta$ as per \Cref{def:zflip}, we have (absorbing $k < 1$ into the bound): 
\begin{equation*}
    \theta_0^3 < \frac{3 r^2 \ln(1/\epsilon)}{2\delta(1 + r)n}.
\end{equation*}
\textbf{Case 1.} If $\delta \geq \delta_c = \frac{6}{\pi^3}\cdot\frac{\ln(1/\epsilon)}{n}$, then
\begin{equation*}
    \theta_0^3 < \frac{\pi^3}{4}\cdot\frac{r^2}{1+r}.
\end{equation*}
Since $\delta \geq \delta_c > 0$, we have $r < 1$, and hence $\frac{r^2}{1+r} < \frac12$. Therefore $\theta_0^3 < \frac{\pi^3}{8}$, so $\theta_0 < \frac{\pi}{2}$.}

With $\theta_0 < \frac{\pi}{2}$, we can apply \Cref{cor:BE3.2} to our modified polynomial $\widetilde{Q}_c$ \rev{(as defined in \Cref{eq:Qtilde_properties}, with $M = 1/c_0$). $\widetilde{Q}_c$ has $|\widetilde{c}_0| = 1$ and $|\widetilde{c}_j| \leq M$ for all $j$, so \Cref{cor:BE3.2} applies}. Note that \rev{$\vert Q_c(z) \vert = c_0 \vert \widetilde{Q}_c(z) \vert > \epsilon \vert \widetilde{Q}_c(z) \vert$}, so:
\begin{align}
    \rev{F_c} &\geq \exp\left(-\frac{1-\nu_2^2}{8\nu_2^2} \theta_0^2n \right) \cdot \exp \left(-O\left(\frac{1}{\nu_1\theta_0}\right) \ln(1/\epsilon)\right) \nonumber \\
    &\geq \exp\left(-\frac{1 - \nu_2^2}{8\nu_2^2} \theta_0^2n -O\left(\frac{1}{\nu_1\theta_0}\right) \ln(1/\epsilon)\right).
\end{align} 
Plugging in the value of $\theta_0$, we get that
\begin{equation}
    \rev{F_c} \geq \exp\left(-O\left(\frac{\ln^{2/3}(1/\epsilon)\cdot (n(1-\nu_2^2))^{1/3}}{(\nu_1\nu_2)^{2/3}}\right)\right). 
\end{equation}
\rev{Plugging in $\nu_1 = \frac{2}{3}$ and $\nu_2 = 1 - \delta$ from \Cref{def:zflip}, we get
\begin{equation}
    \rev{F_c} \geq \exp\left(-O\left((\delta n)^{1/3}\ln^{2/3}(1/\epsilon)\right)\right). 
\end{equation}}

\rev{Thus, we have shown that
\begin{equation}
    \eta(\epsilon) \geq \exp\left(-O\left((\delta n)^{1/3}\ln^{2/3}(1/\epsilon)\right)\right). 
\end{equation}}

\rev{\textbf{Case 2.} $\delta < \delta_c$. In this case, we can just lower bound the maximum by evaluating at $\theta = \frac{\pi}{2}$, which gives us}
\begin{align*}
    F_c &\geq \exp\left(-\frac{1-\nu_2^2}{32\nu_2^2} \pi^2n \right) \cdot \exp \left(-O\left(\frac{2}{\nu_1\pi}\right) \ln(1/\epsilon)\right) \\
    &\geq \exp\left(-O\left(\frac{1}{\nu_1}\right) \ln(1/\epsilon)\right) \\
    &\geq \epsilon^{O\left(\frac{1}{\nu_1}\right)} \\
    &\rev{\geq \epsilon^{O(1)}.} &\tag{\rev{with $\nu_1 = \frac{2}{3}$}}
\end{align*} 

Therefore, we have shown that 
\[ \rev{\eta(\epsilon) \geq
    \begin{cases}
        \exp\left(-O\left((\delta n)^{1/3}\ln^{2/3}(1/\epsilon)\right)\right), & \text{if } \delta \geq \delta_c, \\[2.5ex]
        \epsilon^{O(1)}, & \text{if } \delta < \delta_c.
    \end{cases}} \qedhere \]
\end{proof}

Combining this bound with \Cref{thm:combo} yields \Cref{thm:mainresult}.

\pagebreak
\bibliographystyle{alpha}
\bibliography{references}

@article{FO21paulierror,
  doi = {10.22331/q-2021-09-23-549},
  url = {https://doi.org/10.22331/q-2021-09-23-549},
  title = {Pauli error estimation via {P}opulation {R}ecovery},
  author = {Flammia, Steven and O'Donnell, Ryan},
  journal = {{Quantum}},
  issn = {2521-327X},
  publisher = {{Verein zur F{\"{o}}rderung des Open Access Publizierens in den Quantenwissenschaften}},
  volume = {5},
  pages = {549},
  month = sep,
  year = {2021}
}

@misc{DOS17populationrecovery,
      title={Sharp bounds for population recovery}, 
      author={Anindya De and Ryan O'Donnell and Rocco Servedio},
      year={2017},
      eprint={1703.01474},
      archivePrefix={arXiv},
      primaryClass={cs.DS},
      url={https://arxiv.org/abs/1703.01474}, 
}

@article {BE97,
    AUTHOR = {Borwein, Peter and Erd\'{e}lyi, Tam\'{a}s},
     TITLE = {Littlewood-type problems on subarcs of the unit circle},
   JOURNAL = {Indiana University Mathematics Journal},
    VOLUME = {46},
      YEAR = {1997},
    NUMBER = {4},
     PAGES = {1323--1346},
      ISSN = {0022-2518},
   MRCLASS = {30C10},
  MRNUMBER = {1631600},
MRREVIEWER = {Jay M. Jahangiri},
       DOI = {10.1512/iumj.1997.46.1435},
       URL = {https://doi.org/10.1512/iumj.1997.46.1435},
}

@article{FlammiaWallman,
author = {Flammia, Steven and Wallman, Joel},
title = {Efficient Estimation of {P}auli Channels},
year = {2020},
issue_date = {December 2020},
publisher = {Association for Computing Machinery},
address = {New York, NY, USA},
volume = {1},
number = {1},
url = {https://doi.org/10.1145/3408039},
doi = {10.1145/3408039},
journal = {ACM Transactions on Quantum Computing},
month = dec,
articleno = {3},
numpages = {32}
}

@article{HYF21,
  title = {Fast Estimation of Sparse Quantum Noise},
  author = {Harper, Robin and Yu, Wenjun and Flammia, Steven},
  journal = {PRX Quantum},
  volume = {2},
  issue = {1},
  pages = {010322},
  numpages = {26},
  year = {2021},
  month = {Feb},
  publisher = {American Physical Society},
  doi = {10.1103/PRXQuantum.2.010322},
  url = {https://link.aps.org/doi/10.1103/PRXQuantum.2.010322}
}

@inproceedings{polyanskiy2017sample,
  title={Sample complexity of population recovery},
  author={Polyanskiy, Yury and Suresh, Ananda Theertha and Wu, Yihong},
  booktitle={Conference on Learning Theory},
  pages={1589--1618},
  year={2017},
  organization={PMLR}
}

@article{wallman2016noise,
  title={Noise tailoring for scalable quantum computation via randomized compiling},
  author={Wallman, Joel and Emerson, Joseph},
  journal={Physical Review A},
  volume={94},
  number={5},
  pages={052325},
  year={2016},
  publisher={APS}
}

@article{knill2008randomized,
  title={Randomized benchmarking of quantum gates},
  author={Knill, Emanuel and Leibfried, Dietrich and Reichle, Rolf and Britton, Joe and Blakestad, Brad and Jost, John and Langer, Chris and Ozeri, Roee and Seidelin, Signe and Wineland, David},
  journal={Physical Review A},
  volume={77},
  number={1},
  pages={012307},
  year={2008},
  publisher={APS}
}

@article{bennett1992communication,
  title={Communication via one-and two-particle operators on {E}instein--{P}odolsky--{R}osen states},
  author={Bennett, Charles and Wiesner, Stephen},
  journal={Physical Review Letters},
  volume={69},
  number={20},
  pages={2881},
  year={1992},
  publisher={APS}
}

@article{canonne2020short,
  title={A short note on learning discrete distributions},
  author={Canonne, Cl{\'e}ment},
  journal={arXiv:2002.11457},
  year={2020}
}

@article{wigderson2016population,
  title={Population recovery and partial identification},
  author={Wigderson, Avi and Yehudayoff, Amir},
  journal={Machine Learning},
  volume={102},
  number={1},
  pages={29--56},
  year={2016},
  publisher={Springer}
}

@inproceedings{dvir2012restriction,
  title={Restriction access},
  author={Dvir, Zeev and Rao, Anup and Wigderson, Avi and Yehudayoff, Amir},
  booktitle={Proceedings of the 3rd Innovations in Theoretical Computer Science Conference},
  pages={19--33},
  year={2012}
}

@inproceedings{moitra2013polynomial,
  title={A polynomial time algorithm for lossy population recovery},
  author={Moitra, Ankur and Saks, Michael},
  booktitle={2013 IEEE 54th Annual Symposium on Foundations of Computer Science},
  pages={110--116},
  year={2013},
  organization={IEEE}
}

@article{chen2022quantum,
  title={Quantum advantages for Pauli channel estimation},
  author={Chen, Senrui and Zhou, Sisi and Seif, Alireza and Jiang, Liang},
  journal={Physical Review A},
  volume={105},
  number={3},
  pages={032435},
  year={2022},
  publisher={APS}
}

@article{chen2023learnability,
  title={The learnability of Pauli noise},
  author={Chen, Senrui and Liu, Yunchao and Otten, Matthew and Seif, Alireza and Fefferman, Bill and Jiang, Liang},
  journal={Nature Communications},
  volume={14},
  number={1},
  pages={52},
  year={2023},
  publisher={Nature Publishing Group UK London}
}

@inproceedings{batman2013finding,
  title={Finding heavy hitters from lossy or noisy data},
  author={Batman, Lucia and Impagliazzo, Russell and Murray, Cody and Paturi, Ramamohan},
  booktitle={International Workshop on Approximation Algorithms for Combinatorial Optimization},
  pages={347--362},
  year={2013},
  organization={Springer}
}

@inproceedings{de2016noisy,
  title={Noisy population recovery in polynomial time},
  author={De, Anindya and Saks, Michael and Tang, Sijian},
  booktitle={IEEE 57th Annual Symposium on Foundations of Computer Science (FOCS)},
  pages={675--684},
  year={2016},
  organization={IEEE}
}

@inproceedings{lovett2015improved,
  title={Improved noisy population recovery, and reverse {B}onami--{B}eckner inequality for sparse functions},
  author={Lovett, Shachar and Zhang, Jiapeng},
  booktitle={Proceedings of the forty-seventh annual ACM symposium on Theory of computing},
  pages={137--142},
  year={2015}
}

@inproceedings{lovett2017noisy,
  title={Noisy population recovery from unknown noise},
  author={Lovett, Shachar and Zhang, Jiapeng},
  booktitle={Conference on Learning Theory},
  pages={1417--1431},
  year={2017},
  organization={PMLR}
}

@book{conway1978,
  author    = {Conway, John B.},
  title     = {Functions of One Complex Variable},
  series    = {Graduate Texts in Mathematics},
  volume    = {11},
  edition   = {2},
  publisher = {Springer-Verlag},
  address   = {New York},
  year      = {1978},
}

\end{document}